\documentclass{article}
\usepackage[left=1in, right=1in, top=1in, bottom=1in]{geometry}
\usepackage{enumitem}
\usepackage{authblk}
\usepackage{graphicx}	
\usepackage{amsmath, amsthm, amssymb}
\usepackage{xspace}
\usepackage{comment}
\usepackage{hyperref}
\usepackage{xcolor}
\usepackage{algorithm}
\usepackage{algpseudocode}
\usepackage{lineno}

\newcommand{\algprobm}[1]{\textsc{#1}\xspace}

\newcommand{\Soc}{\text{Soc}}
\newcommand{\Fac}{\text{Fac}}
\newcommand{\F}{\mathbb{F}}
\newcommand{\Z}{\mathbb{Z}}

\theoremstyle{plain}
\newtheorem{theorem}{Theorem}[section]

\newtheorem{corollary}[theorem]{Corollary}
\newtheorem{lemma}[theorem]{Lemma}

\theoremstyle{definition}
\newtheorem{definition}[theorem]{Definition}
\newtheorem{remark}[theorem]{Remark}
\newtheorem{question}[theorem]{Question}

\newcommand{\Lem}[1]{Lem.~\ref{#1}\xspace}
\newcommand{\Cor}[1]{Cor.~\ref{#1}\xspace}

\newcommand{\Thm}[1]{Thm.~\ref{#1}\xspace}
\newcommand{\Def}[1]{Def.~\ref{#1}\xspace}
\newcommand{\Rmk}[1]{Rmk.~\ref{#1}\xspace}

\DeclareMathOperator{\Aut}{Aut}

\DeclareMathOperator{\rad}{Rad}

\DeclareMathOperator{\poly}{poly}
\DeclareMathOperator{\supp}{supp}

\DeclareMathOperator{\rk}{rk}

\title{Count-Free Weisfeiler--Leman and Group Isomorphism\footnote{NAC was partially supported by J. A. Grochow's startup funds. ML was partially supported by J. A. Grochow's NSF award CISE-2047756 and the University of Colorado Boulder, Department of Computer Science Summer Research Fellowship. We thank J.A. Grochow and the anonymous referee for their helpful feedback.} }

\author[1]{Nathaniel A. Collins}
\author[2]{Michael Levet}
\affil[1]{Department of Mathematics, Colorado State University}
\affil[2]{Department of Computer Science, College of Charleston}

\begin{document}
\maketitle

\begin{abstract}
We investigate the power of counting in \algprobm{Group Isomorphism}. We first leverage the count-free variant of the Weisfeiler--Leman Version I algorithm for groups (Brachter \& Schweitzer, LICS 2020) in tandem with bounded non-determinism and limited counting to improve the parallel complexity of isomorphism testing for several families of groups. These families include:
\begin{itemize}
\item Direct products of non-Abelian simple groups.

\item Coprime extensions, where the normal Hall subgroup is Abelian and the complement is an $O(1)$-generated solvable group with solvability class $\poly \log \log n$. This notably includes instances where the complement is an $O(1)$-generated nilpotent group. This problem was previously known to be in $\textsf{P}$ (Qiao, Sarma, \& Tang, STACS 2011), and the complexity was recently improved to $\textsf{L}$ (Grochow \& Levet, FCT 2023). 

\item Graphical groups of class $2$ and exponent $p > 2$ (Mekler, \textit{J. Symb. Log.}, 1981) arising from the CFI and twisted CFI graphs (Cai, F\"urer, \& Immerman, \textit{Combinatorica} 1992) respectively. In particular, our work improves upon previous results of Brachter \& Schweitzer (LICS 2020).
\end{itemize}

Notably, each of these families were previously known to be identified by the \textit{counting} variant of the more powerful Weisfeiler--Leman Version II algorithm.  We finally show that the $q$-ary count-free pebble game is unable to even distinguish Abelian groups. This extends the result of Grochow \& Levet (ibid), who established the result in the case of $q = 1$. The general theme is that some counting appears necessary to place $\algprobm{Group Isomorphism}$ into $\textsf{P}$.
\end{abstract}

\thispagestyle{empty}

\newpage

\setcounter{page}{1}

\section{Introduction}
\label{sec:introduction}
The \algprobm{Group Isomorphism} problem (\algprobm{GpI}) takes as input two finite groups $G$ and $H$, and asks if there exists an isomorphism $\varphi : G \to H$. Here, we will consider groups given by their multiplication (a.k.a. Cayley) tables. The generator-enumerator algorithm, attributed to Tarjan in 1978 \cite{MillerTarjan}, has time complexity $n^{\log_{p}(n) + O(1)}$, where $n$ is the order of the group and $p$ is the smallest prime dividing $n$. In more than 40 years, this bound has escaped largely unscathed: Rosenbaum \cite{Rosenbaum2013BidirectionalCD} (see \cite[Sec. 2.2]{GR16}) improved this to $n^{(1/4)\log_p(n) + O(1)}$. And even the impressive body of work on practical algorithms for this problem, led by Eick, Holt, Leedham-Green and O'Brien (e.\,g., \cite{BEO02, ELGO02, BE99, CH03}) still results in an $n^{\Theta(\log n)}$-time algorithm in the general case (see \cite[Page 2]{WilsonSubgroupProfiles}). In the past several years, there have been significant advances on algorithms with worst-case guarantees on the serial runtime for special cases of this problem including Abelian groups \cite{Kavitha, Vikas, Savage}, direct product decompositions \cite{WilsonDirectProductsArxiv, KayalNezhmetdinov}, groups with no Abelian normal subgroups \cite{BCGQ, BCQ}, coprime and tame group extensions \cite{Gal09,QST11,BQ, GQ15}, low-genus $p$-groups and their quotients \cite{LW12,BMWGenus2, IvanyosQ19, SchrockThesis}, Hamiltonian groups \cite{DasSharma}, and groups of almost all orders \cite{DietrichWilson}.

Despite decades of research, there are no known complexity-theoretic lower bounds against \algprobm{Group Isomorphism}, and it is unclear how to even establish such bounds. There is reason for this-- the work of Chattopadhyay, Tor\'an, \& Wagner rules out existing techniques for establishing lower bounds (cf., \cite{VollmerText}) against very weak complexity classes like $\textsf{AC}^{0}$, which is in fact a \textit{proper} subclass of $\textsf{P}$. In particular, this rules out $\textsf{AC}^{0}$-reductions (as well as several stronger notions of parallel reduction) from \algprobm{Majority} (counting) and \algprobm{Parity} to \algprobm{Group Isomorphism}. Note that many $\textsf{NP}$-completeness reductions are $\textsf{AC}^{0}$-computable.

Even though \algprobm{Majority} does not reduce to \algprobm{Group Isomorphism} (under $\textsf{AC}^{0}$-computable reductions), counting has nonetheless been a key ingredient for advances in \algprobm{Group Isomorphism}. For instance, the Fundamental Theorem of finitely generated Abelian groups provides that two finite Abelian groups $G$ and $H$ of order $n$ are isomorphic if and only if for each divisor $d \mid n$, the number of elements of order $d$ in $G$ is the same as in $H$. Lipton, Snyder, \& Zalcstein \cite{LiptonSnyderZalcstein} leveraged this to establish the first polynomial-time isomorphism test for Abelian groups, and subsequent improvements in both the serial runtime \cite{Vikas, Savage, Kavitha} and parallel (circuit) complexity \cite{ChattopadhyayToranWagner} have crucially relied on the Fundamental Theorem of finitely generated Abelian groups. The Remak--Krull--Schmidt Theorem similarly provides that a group is determined up to isomorphism by the isomorphism classes of its indecomposable direct factors and their multiplicities. This has been leveraged crucially in algorithms that rely on determining whether two group actions are equivalent \cite{Gal09,QST11,BQ, GQ15, GQcoho}. Recent works utilizing the Weisfeiler--Leman algorithm for \algprobm{Group Isomorphism} also crucially rely on counting \cite{QiaoLiWL, BGLQW, WLGroups, BrachterSchweitzerWLLibrary, GLWL1, GLDescriptiveComplexity}.

By treating \textit{both} counting and nondeterminism as scarce resources, Grochow \& Levet \cite{GLWL1} recently obtained further improvements in the parallel complexity of isomorphism testing for Abelian groups. In light of this result, as well as the notable lack of lower bounds, it is natural to inquire as to the extent that counting is necessary to place \algprobm{Group Isomorphism} into $\textsf{P}$. Indeed, the circuit model is well-suited to investigate this question, as we can analyze the number of counting (\textsf{Majority}) gates that are present. We will leverage this technique of Grochow \& Levet to make further advances in the parallel complexity of \algprobm{Group Isomorphism}. This brings us to the Weisfeiler--Leman procedure, which will serve as the backbone of our algorithmic approach.

The Weisfeiler--Leman algorithm is a key combinatorial subroutine that has driven advances for the \algprobm{Graph Isomorphism} problem for several decades (see Section~\ref{sec:RelatedWork} for detailed discussion, as well as Sandra Kiefer's dissertation for a very thorough and current survey of the Weisfeiler--Leman algorithm \cite{KieferThesis}). The \algprobm{Group Isomorphism} problem, when the groups are given by their multiplication tables, reduces to \algprobm{Graph Isomorphism} \cite{MillerTarjan} (see as well the reduction used in Weisfeiler--Leman Version III \cite{WLGroups}). As a result, \algprobm{Group Isomorphism} can be viewed as a special algebraic subproblem of \algprobm{Graph Isomorphism}. It is thus natural to investigate the power of Weisfeiler--Leman in the setting of groups.

One approach for isomorphism testing is to identify an isomorphism invariant relation that distinguishes two objects whenever they are non-isomorphic. Indeed, this is intuitively the aim of Weisfeiler--Leman. The original formulation of Weisfeiler--Leman (what is known today as the $2$-dimensional Weisfeiler--Leman) establishes an equivalence between the algorithmic procedure and binary relational structures known as \textit{coherent configurations} \cite{WLOriginal, Weisfeiler1976OnCA}, which are a staple in algebraic graph theory. Babai \& Mathon \cite{BabaiMathon} generalized Weisfeiler--Leman for higher-dimensions, establishing an equivalence with $k$-ary relational structures that they call \textit{generalized coherent configurations}. 

In order to construct these relations, the $k$-dimensional Weisfeiler--Leman algorithm iteratively colors $k$-tuples of elements\footnote{In our setting, these will be $k$-tuples of group elements, but the procedure can be applied to any relational structure such as monoids or graphs.} in an isomorphism invariant manner. If at the end of a given iteration, the multiset of colors present for one object differs from the other, we can conclude that the two objects are not isomorphic. It is open whether there exists a constant $k$, where $k$-dimensional Weisfeiler--Leman distinguishes any pair of non-isomorphic groups. A priori, it may seem unnatural to consider a coloring procedure that, on the surface, appears to ignore key group theoretic structure. This is, however, far from the case. Immerman \& Lander \cite{ImmermanLander1990} (see as well \cite{CFI}) independently generalized Weisfeiler--Leman for higher-dimensions, establishing equivalences with tools from finite model theory including Ehrenfeucht--Fra\"iss\'e pebble games and first-order logic with counting quantifiers ($\textsf{FO} + \textsf{C}$). The pebble game characterization allows us to investigate at a granular level, the extent to which the Weisfeiler--Leman algorithm detects the relevant group theoretic structures and the number of iterations required to do so, essentially marrying group theory, finite model theory, and combinatorics to make advances in computation.

There is a relaxation of the Weisfeiler--Leman algorithm to a \textit{count-free} variant, that intuitively considers the \textit{set} of colors at each round, rather than the full multiset of colors. Count-free Weisfeiler--Leman is equivalent to first-order logic \textit{without} counting quantifiers ($\textsf{FO}$) \cite{ImmermanLander1990, CFI}, and also admits an efficient parallel implementation using circuits \textit{without} counting ($\textsf{Majority}$) gates\footnote{The standard counting variant of Weisfeiler--Leman admits an analogous parallel implementation that utilizes a polynomial number of $\textsf{Majority}$ gates.} \cite{GroheVerbitsky}. As we will fruitfully leverage count-free Weisfeiler--Leman in tandem with bounded non-determinism and limited counting, our results suggest, both from the perspective of logics and circuits, that perhaps a limited amount of counting might suffice to place \algprobm{Group Isomorphism} into $\textsf{P}$. This is in stark contrast to the setting of \algprobm{Graph Isomorphism}, where $\textsf{FO} + \textsf{C}$ notoriously fails to yield \textit{any} improvements in the general case. In fact, there exists an infinite family $(G_{n}, H_{n})_{n\in \mathbb{Z}^{+}}$ of pairs of non-isomorphic graphs, where the standard counting Weisfeiler--Leman algorithm requires time $n^{\Theta(n)}$ to distinguish $G_n$ from $H_n$ \cite{CFI}, which is worse than the brute-force approach of enumerating all possible $n!$ bijections.

\noindent \\ \textbf{Main Results.} In this paper, we show that the count-free Weisfeiler--Leman (WL) Version I algorithm serves as a key subroutine in developing efficient parallel isomorphism tests for several families of groups. Namely, we will establish the following. 

\begin{theorem} \label{thm:Main}
For the following classes of groups, the isomorphism problem can be efficiently parallelized with improved complexity using the count-free variant of Weisfeiler--Leman Version I:
\begin{enumerate}[label=(\alph*)]
\item Direct products of non-Abelian simple groups;
\item Coprime extensions $H \ltimes N$, where the normal Hall subgroup $N$ is Abelian, and the complement $H$ is an $O(1)$-generated solvable group with solvability class $\poly \log \log n$; and
\item Graphical groups of class $2$ and exponent $p > 2$, arising from the CFI and twisted CFI graphs \cite{CFI}.
\end{enumerate}
\end{theorem}

\noindent \\ Our strategy in proving \Thm{thm:Main} is similar to that of Grochow \& Levet \cite{GLWL1} for handling Abelian groups. We will begin by briefly recalling their strategy here. Let $G$ and $H$ be our groups. Grochow \& Levet first utilized $O(\log \log n)$ rounds of the constant-dimensional count-free Weisfeiler--Leman Version I algorithm (which is $\textsf{FOLL}$ computable \cite{GroheVerbitsky}). This sufficed to distinguish Abelian groups from non-Abelian groups, as well as to distinguish two group elements with different orders. Now if both $G$ and $H$ are Abelian, the two groups are isomorphic if and only if for each divisor $d$ of $|G|$, $G$ and $H$ have the same number of elements of order $d$. Thus, if $G \not \cong H$, there exists an order-- and hence, a color class-- where the multiplicity of said color class $C$ is greater for $G$ than for $H$. While count-free Weisfeiler--Leman is not powerful enough to detect this difference, we may use $O(\log n)$ non-deterministic bits to guess such a color class. For tuples belonging to $C$ in $G$, we feed a $1$ to our $\textsf{Majority}$ gate; and for the corresponding tuples in $H$, we feed a $0$ to our $\textsf{Majority}$ gate. This last step is $\textsf{MAC}^{0}$-computable (that is, computable by an $\textsf{AC}^{0}$ circuit with a single $\textsf{Majortiy}$ gate at the output). In total, this yields an upper bound of $\beta_{1}\textsf{MAC}^{0}(\textsf{FOLL})$, where the $\beta_{1}$ indicates the $O(\log n)$ non-deterministic bits (see Section~\ref{sec:Complexity} and Remark~\ref{rmk:Diagram} for more discussion on how $\beta_{1}\textsf{MAC}^{0}(\textsf{FOLL})$ compares to other complexity classes). This notably improved upon the upper bound of $\textsf{TC}^{0}(\textsf{FOLL})$ due to Chattopadhyay, Tor\'an, \& Wagner \cite{ChattopadhyayToranWagner}.

Instead of order, we will show that there exist different invariants where, after running $\poly \log \log n$ rounds of the constant-dimensional count-free WL Version I algorithm, the multiset of colors differ between our two input groups $G$ and $H$, whenever $G \not \cong H$. The strategy is otherwise identical to that of Grochow \& Levet for Abelian groups. We also note that for direct products of non-Abelian simple groups and the graphical groups in question, $O(\log \log n)$ rounds of count-free WL suffice. As a consequence, we obtain upper bounds of $\beta_{1}\textsf{MAC}^{0}(\textsf{FOLL})$ for isomorphism testing of these families of groups, essentially\footnote{The isomorphism problem for Abelian groups is also known to belong to $\textsf{L}$ \cite{ChattopadhyayToranWagner}, which is not known for the groups considered in \Thm{thm:Main}(c).} matching that for Abelian groups. This is particularly surprising for \Thm{thm:Main}(c), as (i) the CFI graphs serve as hard instances for WL \cite{CFI}, and (ii) class $2$ $p$-groups of exponent $p$ have long been acknowledged to be a bottleneck case of \algprobm{Group Isomorphism} (see e.g., \cite{BCGQ, DietrichWilson, GrochowQiaoTensors}), while Abelian groups are amongst the \textit{easiest} cases.

Each of the families in \Thm{thm:Main} were previously known to be identifiable by the counting variant of the more powerful Weisfeiler--Leman Version II algorithm (see Sections~\ref{sec:WLPrelims} and \ref{sec:ParallelWL} for thorough discussion on the technical differences between Weisfeiler--Leman Versions I and II). 
\begin{enumerate}[label=(\alph*)]
\item Brachter \& Schweitzer \cite{BrachterSchweitzerWLLibrary} established this result for direct products of non-Abelian simple groups, and a careful analysis of their work shows that only $O(1)$ rounds suffice. 

\item Grochow \& Levet \cite{GLWL1} showed that $O(1)$-dimensional WL Version II identifies in $O(1)$ rounds the class of coprime extensions $H \ltimes N$, where $N$ is Abelian and $H$ is $O(1)$-generated (with no restriction on the solvability class of $H$). While count-free WL Version I can also identify $O(1)$-generated groups in $O(\log n)$ rounds, it is unclear whether the number of iterations can be improved. This places us in a complexity class ($\textsf{AC}^{1}$), which in particular can handle counting (both \algprobm{Parity} and \algprobm{Majority} belong to $\textsf{AC}^{1}$). 

By considering $H$ to be solvable with solvability class $\poly \log \log n$, the count-free WL Version I algorithm can identify $H$ in $\poly \log \log n$ rounds, allowing us to remain in a complexity class that cannot compute \algprobm{Parity}. Here, we crucially leverage the techniques established by Barrington, Kadau, Lange, \& McKenzie \cite{BKLM} in their work on the \algprobm{Cayley Group Membership} problem. Determining whether either the counting or count-free variant of WL Version I identifies $O(1)$-generated groups in $O(\log \log n)$ rounds is closely related to determining whether the \algprobm{Cayley Group Membership} problem belongs to the complexity class $\textsf{FOLL}$, the latter of which is a longstanding open problem \cite{BKLM}. Nonetheless, the assumption that $H$ has solvability class $\poly \log \log n$ is already quite robust, as this includes all finite nilpotent groups, which have solvability class $O(\log \log n)$ \cite{Therien1980ClassificationOR}.

\item The graphical groups of class $2$ and exponent $p > 2$ arising from the CFI graphs were first considered by Brachter \& Schweitzer \cite{WLGroups}, who showed that $3$-WL Version II suffices. A careful analysis of their work shows that $O(\log n)$ rounds suffice. Recall that the CFI graphs serve as hard instances for Weisfeiler--Leman in the setting of graphs \cite{CFI}. Furthermore, the CFI graphs have also been fruitfully leveraged to establish lower bounds against generalizations of Weisfeiler--Leman in the direction of logics, including \textsf{Rank Logic} \cite{LichterRankLogic} and restrictions of \textsf{Choiceless Polynomial Time} \cite{PagoCPT}. The fact that the \textit{count-free} WL, in tandem with bounded non-determinism and a single \textsf{Majority} gate, can distinguish between graphical groups arising from the CFI and twisted CFI graphs further illustrates the power of WL in the setting of groups.
\end{enumerate}

\noindent 
\begin{remark}
From a group theoretic perspective, we find it surprising that Weisfeiler--Leman distinguishes the graphical groups arising from the CFI and twisted CFI graphs. A class $2$ $p$-group $G$ of exponent $p > 2$ decomposes as an extension of the center $Z(G)$, with an elementary Abelian quotient $G/Z(G)$. While it is easy to check whether two such groups $G$ and $H$ satisfy $Z(G) \cong Z(H)$ and $G/Z(G) \cong H/Z(H)$, the difficulty is deciding whether $Z(G)$ and $G/Z(G)$ are \textit{glued together} in a compatible manner as $Z(H)$ and $H/Z(H)$. For class $2$ $p$-groups of exponent $p > 2$, this is the \algprobm{Cohomology Class Isomorphism} problem (see \cite{GQcoho} for the precise formulation of \algprobm{Cohomology Class Isomorphism}). The usual approach of dealing with cohomology algorithmically is to leverage tools from (multi)linear algebra-- see for instance, \cite{GQcoho, BMWGenus2, LewisWilson, IvanyosQ19, SchrockThesis, GrochowQiaoTensors}. In the setting of \Thm{thm:Main}(c), the CFI graphs provide sufficient structure to restrict the cohomology, allowing Weisfeiler--Leman to solve the isomorphism problem and avoiding the need for tools from (multi)linear algebra. 

The techniques we employ to prove \Thm{thm:Main}(c), which extend the work of Brachter \& Schweitzer \cite{WLGroups}, do not appear to extend to arbitrary graphical groups. For instance, \cite[Lemma~4.14]{WLGroups} provides that the vertex set of a tree does not yield a canonical generating set for the corresponding graphical group $G$, as there is a $\varphi \in \Aut(G)$ mapping (the group element corresponding to) a leaf node $v$ to (the group element corresponding to) the edge incident to $v$. Trees are considered to be easy cases for \algprobm{Graph Isomorphism}:  they are identified by $1$-dimensional Weisfeiler--Leman \cite{edmonds_1965, ImmermanLander1990}, and the isomorphism problem for trees is $\textsf{L}$-complete \cite{ElberfeldSchweitzer} (and so can be efficiently parallelized). The CFI graphs are considerably more complicated than trees, and group theoretic methods appear necessary to identify them in polynomial-time \cite{LUKS198242}. It is a longstanding open problem whether the CFI graphs admit an efficient parallel isomorphism test.
\end{remark}

We now turn to considering a higher-arity version of the count-free pebble game (where the $1$-ary case was introduced for graphs \cite{ImmermanLander1990, CFI} and adapted for groups \cite{GLWL1}). In the $q$-ary game, Spoiler can pebble up to $q$ elements at a single round. A priori, it seems plausible that the higher-arity count-free game might be able to make progress on \algprobm{Group Isomorphism}. Nonetheless, we show that the higher-arity count-free game is unable to even distinguish Abelian groups. 

\begin{theorem} \label{thm:LowerBound}
Let $q \geq 1, n \geq 5$. Let $G_{n} := (\mathbb{Z}/2\mathbb{Z})^{qn} \times (\mathbb{Z}/4\mathbb{Z})^{qn}$ and $H_{n} := (\mathbb{Z}/2\mathbb{Z})^{q(n-2)} \times (\mathbb{Z}/4\mathbb{Z})^{q(n+1)}$. Duplicator has a winning strategy in the $q$-ary, $qn/4$-pebble game.
\end{theorem}

\begin{remark}
The case of $q = 1$ corresponds to count-free Weisfeiler--Leman as well as $\textsf{FO}$ \cite{CFI, GLWL1}, and the lower bound in the $q = 1$ case was established by Grochow \& Levet \cite[Theorem~7.10]{GLWL1}. Here, \Thm{thm:LowerBound} suggests some counting might be necessary to place $\algprobm{GpI}$ into $\textsf{P}$. In particular, \Thm{thm:LowerBound} provides additional evidence that $\algprobm{GpI}$ might not belong to (uniform) $\textsf{AC}^{0}$. In contrast, while some counting appears necessary, \Thm{thm:Main} and \cite[Theorem~7.15]{GLWL1} provide evidence that the full power of $\textsf{FO} + \textsf{C}$ might not be required to resolve $\algprobm{GpI}$.

We also contrast the higher-arity count-free game with the analogous game in the counting setting. The $3$-ary counting game is trivially able to resolve $\algprobm{Group Isomorphism}$. Namely, Duplicator chooses a bijection $f : G \to H$. If $f$ is not an isomorphism, then Spoiler pebbles a triple $(g_{1}, g_{2}, g_{1}g_{2})$ where $f(g_{1})f(g_{2}) \neq f(g_{1}g_{2})$. The power of the $2$-ary counting game remains an intriguing open question in the setting of groups. Grochow \& Levet \cite{GLDescriptiveComplexity} showed that this pebble game identifies groups without Abelian normal subgroups using only $O(1)$ pebbles and $O(1)$ rounds. This family of groups has rich structure and is highly non-trivial. In the Cayley table model, it took a series of two papers \cite{BCGQ, BCQ} to obtain a polynomial-time isomorphism test for this family of groups.
\end{remark}

\section{Preliminaries}

\subsection{Groups}

Unless stated otherwise, all groups are assumed to be finite and represented by their Cayley tables. 
For a group of order $n$, the Cayley table has $n^{2}$ entries, each represented by a binary string of size $\lceil \log_{2}(n) \rceil$. For an element $g$ in the group $G$, we denote the \textit{order} of $g$ by $|g|$. We use $d(G)$ to denote the minimum size of a generating set for the group $G$. 

We say that a normal subgroup $N \trianglelefteq G$ \textit{splits} in $G$ if there exists a subgroup $H \leq G$ such that $H \cap N = \{1\}$ and $G = HN$. The conjugation action of $H$ on $N$ allows us to express multiplication of $G$ in terms of pairs $(h, n) \in H \times N$. We note that the conjugation action of $H$ on $N$ induces a group homomorphism $\theta : H \to \Aut(N)$ mapping $h \mapsto \theta_{h}$, where $\theta_{h} : N \to N$ sends $\theta_{h}(n) = hnh^{-1}$. So given $(H, N, \theta)$, we may define the group $H \ltimes_{\theta} N$ on the set $\{ (h, n) : h \in H, n \in N \}$ with the product $(h_{1}, n_{1})(h_{2}, n_{2}) = (h_{1}h_{2}, \theta_{h_{2}^{-1}}(n_{1})n_{2})$. We refer to the decomposition $G = H \ltimes_{\theta} N$ as a \textit{semidirect product} demoposition. When the action $\theta$ is understood, we simply write $G = H \ltimes N$. 

A \textit{normal Hall subgroup} $N \trianglelefteq G$ is a subgroup such that $|N|$ and $[G : N] = |G/N|$ are relatively prime. We are particularly interested in semidirect products when $N$ is a normal Hall subgroup. To this end, we recall the Schur--Zassenhaus Theorem \cite[(9.1.2)]{Robinson1982}.

\begin{theorem}[Schur--Zassenhaus]
Let $G$ be a finite group of order $n$, and let $N$ be a normal Hall subgroup. Then there exists a complement $H \leq G$, such that $\text{gcd}(|H|, |N|) = 1$ and $G = H \ltimes N$. Furthermore, if $H$ and $K$ are complements of $N$, then $H$ and $K$ are conjugate.
\end{theorem}

The \textit{Frattini} subgroup $\Phi(G)$ is the intersection of maximal subgroups. For a class $2$ $p$-group of exponent $p > 2$, $G/\Phi(G) \cong (\mathbb{Z}/p\mathbb{Z})^{d(G)}$.

For a group $G$, $G' := [G, G]$ denotes the commutator subgroup. The \textit{derived series} of a group $G$, indexed $G^{(\alpha)}$, is defined as follows: $G^{(0)} := G$, and for $\alpha > 0$, $G^{(\alpha)} := [G^{(\alpha-1)}, G^{(\alpha-1)}]$. A group is said to be \textit{solvable} if there exists an $\ell < \infty$ such that $G^{(\ell)} = 1$. We refer to the minimum such $\ell$ as the \textit{solvability class} of $G$.

\subsection{Complexity Classes} \label{sec:Complexity}
We assume familiarity with the complexity classes $\textsf{P}, \textsf{NP}, \textsf{L}$, and $\textsf{NL}$-- we refer the reader to standard references \cite{ComplexityZoo, AroraBarak}. We will now turn to introducing key notions from circuit complexity. For a standard reference, see \cite{VollmerText}. We consider Boolean circuits using the \textsf{AND}, \textsf{OR}, \textsf{NOT}, and \textsf{Majority} gates, where $\textsf{Majority}(x_{1}, \ldots, x_{n}) = 1$ if and only if $> n/2$ of the inputs are $1$. Otherwise, $\textsf{Majority}(x_{1}, \ldots, x_{n}) = 0$. We use $\textsf{Majority}$ to denote a gate in a circuit, and \algprobm{Majority} to denote the algorithmic problem of deciding, when given input bits $x_{1}, \ldots, x_{n}$, whether $> n/2$ of the bits are $1$'s.

\begin{definition}
Fix $k \geq 0$. We say that a language $L$ belongs to (uniform) $\textsf{NC}^{k}$ if there exist a (uniform) family of circuits $(C_{n})_{n \in \mathbb{N}}$ over the $\textsf{AND}, \textsf{OR}, \textsf{NOT}$ gates such that the following hold:
\begin{itemize}
\item The $\textsf{AND}$ and $\textsf{OR}$ gates take exactly $2$ inputs. That is, they have fan-in $2$.
\item $C_{n}$ has depth $O(\log^{k} n)$ and uses (has size) $n^{O(1)}$ gates. Here, the implicit constants in the circuit depth and size depend only on $L$.

\item $x \in L$ if and only if $C_{|x|}(x) = 1$. 
\end{itemize}
\end{definition}

\noindent The complexity class $\textsf{AC}^{k}$ is defined analogously as $\textsf{NC}^{k}$, except that the $\textsf{AND}, \textsf{OR}$ gates are permitted to have unbounded fan-in. That is, a single $\textsf{AND}$ gate can compute an arbitrary conjunction, and a single $\textsf{OR}$ gate can compute an arbitrary disjunction. The complexity class $\textsf{TC}^{k}$ is defined analogously as $\textsf{AC}^{k}$, except that our circuits are now permitted $\textsf{Majority}$ gates of unbounded fan-in.

For every $k$, the following containments are well-known:
\[
\textsf{NC}^{k} \subseteq \textsf{AC}^{k} \subseteq \textsf{TC}^{k} \subseteq \textsf{NC}^{k+1}.
\]

\noindent In the case of $k = 0$, we have that:
\[
\textsf{NC}^{0} \subsetneq \textsf{AC}^{0} \subsetneq \textsf{TC}^{0} \subseteq \textsf{NC}^{1} \subseteq \textsf{L} \subseteq \textsf{NL} \subseteq \textsf{AC}^{1}.
\]

\noindent We note that functions that are $\textsf{NC}^{0}$-computable can only depend on a bounded number of input bits. Thus, $\textsf{NC}^{0}$ is unable to compute the $\textsf{AND}$ function. It is a classical result that $\textsf{AC}^{0}$ is unable to compute \algprobm{Parity} \cite{Smolensky87algebraicmethods}. The containment $\textsf{TC}^{0} \subseteq \textsf{NC}^{1}$ (and hence, $\textsf{TC}^{k} \subseteq \textsf{NC}^{k+1}$) follows from the fact that $\textsf{NC}^{1}$ can simulate the \textsf{Majority} gate.

Let $d(n)$ be a function. The complexity class $\textsf{FO}(d(n))$ is the set of languages decidable by uniform circuit families with \textsf{AND}, \textsf{OR}, and \textsf{NOT} gates of depth $O(d(n))$, polynomial size, and unbounded fan-in. In particular, $\textsf{FOLL} = \textsf{FO}(\log \log n)$. It is known that $\textsf{AC}^{0} \subsetneq \textsf{FOLL} \subsetneq \textsf{AC}^{1}$, and it is open as to whether $\textsf{FOLL}$ is contained in $\textsf{NL}$ and vice-versa \cite{BKLM}. 

The complexity class $\textsf{MAC}^{0}$ is the set of languages decidable by constant-depth uniform circuit familes with a polynomial number of \textsf{AND}, \textsf{OR}, and \textsf{NOT} gates, and at most one $\textsf{Majority}$ gate at the output. The class $\textsf{MAC}^{0}$ was introduced (but not so named) in \cite{VotingPolynomials}, where it was shown that $\textsf{MAC}^{0} \subsetneq \textsf{TC}^{0}$. This class was subsequently given the name $\textsf{MAC}^{0}$ in \cite{LearnabilityAC0}.

For a complexity class $\mathcal{C}$, we define $\beta_{i}\mathcal{C}$ to be the set of languages $L$ such that there exists an $L' \in \mathcal{C}$ such that $x \in L$ if and only if there exists $y$ of length at most $O(\log^{i} |x|)$ such that $(x, y) \in L'$. For any $i, c \geq 0$, $\beta_{i}\textsf{FO}((\log \log n)^{c})$ cannot compute \algprobm{Parity} \cite{ChattopadhyayToranWagner}.

For complexity classes $\mathcal{C}_{1}, \mathcal{C}_{2}$, the complexity class $\mathcal{C}_{1}(\mathcal{C}_{2})$ is the class of $h = g \circ f$, where $f$ is $\mathcal{C}_{1}$-computable and $g$ is $\mathcal{C}_{2}$-computable. For instance, $\beta_{1}\textsf{MAC}^{0}(\textsf{FOLL})$ is the set of functions $h = g \circ f$, where $f$ is $\textsf{FOLL}$-computable and $g$ is $\beta_{1}\textsf{MAC}^{0}$-computable.

\begin{remark} \label{rmk:Diagram}
By simulating the $\poly(n)$ possibilities for the $O(\log n)$ non-deterministic bits, we obtain that $\beta_{1}\textsf{MAC}^{0}$ is contained in $\textsf{TC}^{0}$. Note that this simulation requires $\poly(n)$ $\textsf{Majority}$ gates. So $\beta_{1}\textsf{MAC}^{0}(\textsf{FOLL}) \subseteq \textsf{TC}^{0}(\textsf{FOLL})$. Note that $\textsf{TC}^{0}(\textsf{FOLL})$ improves upon using an arbitrary $\textsf{TC}$ circuit of depth $O(\log \log n)$ and polynomial-size, as the $\textsf{Majority}$ gates only occur in the last $O(1)$ levels of the circuit. Consequently, $\textsf{TC}^{0}(\textsf{FOLL}) \subseteq \textsf{AC}^{1}$. However, $\beta_{1}\textsf{MAC}^{0}(\textsf{FOLL})$ really sits closer to $\textsf{FOLL}$ than to $\textsf{AC}^{1}$.
\end{remark}

\subsection{Weisfeiler--Leman} \label{sec:WLPrelims}

We begin by recalling the Weisfeiler--Leman algorithm for graphs, which computes an isomorphism-invariant coloring. Let $\Gamma$ be a graph, and let $k \geq 2$ be an integer. The $k$-dimensional Weisfeiler--Leman, or $k$-WL, algorithm begins by constructing an initial coloring $\chi_{0} : V(\Gamma)^{k} \to \mathcal{K}$, where $\mathcal{K}$ is our set of colors, by assigning each $k$-tuple a color based on its isomorphism type. That is, two $k$-tuples $(v_{1}, \ldots, v_{k})$ and $(u_{1}, \ldots, u_{k})$ receive the same color under $\chi_{0}$ iff the map $v_i \mapsto u_i$ (for all $i \in [k]$) is an isomorphism of the induced subgraphs $\Gamma[\{ v_{1}, \ldots, v_{k}\}]$ and $\Gamma[\{u_{1}, \ldots, u_{k}\}]$ and for all $i, j$, $v_i = v_j \Leftrightarrow u_i = u_j$. 

For $r \geq 0$, the coloring computed at the $r$th iteration of Weisfeiler--Leman is refined as follows. For a $k$-tuple $\overline{v} = (v_{1}, \ldots, v_{k})$ and a vertex $x \in V(\Gamma)$, define
\[
\overline{v}(v_{i}/x) = (v_{1}, \ldots, v_{i-1}, x, v_{i+1}, \ldots, v_{k}).
\]

The coloring computed at the $(r+1)$st iteration, denoted $\chi_{r+1}$, stores the color of the given $k$-tuple $\overline{v}$ at the $r$th iteration, as well as the colors under $\chi_{r}$ of the $k$-tuples obtained by substituting a single vertex in $\overline{v}$ for another vertex $x$. We examine this multiset of colors over all such vertices $x$. This is formalized as follows:
\begin{align*}
\chi_{r+1}(\overline{v}) = &( \chi_{r}(\overline{v}), \{\!\!\{ ( \chi_{r}(\overline{v}(v_{1}/x)), \ldots, \chi_{r}(\overline{v}(v_{k}/x) ) \bigr| x \in V(\Gamma) \}\!\!\} ),
\end{align*}
where $\{\!\!\{ \cdot \}\!\!\}$ denotes a multiset.

The \textit{count-free} variant of WL considers the set rather than the multiset of colors at each round. Precisely:
\begin{align*}
\chi_{r+1}(\overline{v}) = &( \chi_{r}(\overline{v}), \{ ( \chi_{r}(\overline{v}(v_{1}/x)), \ldots, \chi_{r}(\overline{v}(v_{k}/x) ) \bigr| x \in V(\Gamma) \} ).
\end{align*}

\noindent Note that the coloring $\chi_{r}$ computed at iteration $r$ induces a partition of $V(\Gamma)^{k}$ into color classes. The Weisfeiler--Leman algorithm terminates when this partition is not refined, that is, when the partition induced by $\chi_{r+1}$ is identical to that induced by $\chi_{r}$. The final coloring is referred to as the \textit{stable coloring}, which we denote $\chi_{\infty} := \chi_{r}$.

Brachter \& Schweitzer introduced three variants of WL for groups. WL Versions I and II are both executed directly on the groups, where $k$-tuples of group elements are initially colored. For WL Version I, two $k$-tuples $(g_{1}, \ldots, g_{k})$ and $(h_{1}, \ldots, h_{k})$ receive the same initial color iff (a) for all $i, j, \ell \in [k]$, $g_{i}g_{j} = g_{\ell} \iff h_{i}h_{j} = h_{\ell}$, and (b) for all $i, j \in [k]$, $g_{i} = g_{j} \iff h_{i} = h_{j}$. For WL Version II, $(g_{1}, \ldots, g_{k})$ and $(h_{1}, \ldots, h_{k})$ receive the same initial color iff the map $g_{i} \mapsto h_{i}$ for all $i \in [k]$ extends to an isomorphism of the generated subgroups $\langle g_{1}, \ldots, g_{k} \rangle$ and $\langle h_{1}, \ldots, h_{k} \rangle$. For both WL Versions I and II, refinement is performed in the classical manner as for graphs. Namely, for a given $k$-tuple $\overline{g}$ of group elements,
\begin{align*}
\chi_{r+1}(\overline{g}) = &( \chi_{r}(\overline{g}), \{\!\!\{ ( \chi_{r}(\overline{g}(g_{1}/x)), \ldots, \chi_{r}(\overline{g}(g_{k}/x) ) \bigr| x \in G \}\!\!\} ).
\end{align*}

\noindent We will not use WL Version III, and so we refer the reader to \cite{WLGroups} for details. 

\subsection{Pebbling Game}

We recall the bijective pebble game introduced by \cite{Hella1989, Hella1993} for WL on graphs. This game is often used to show that two graphs $X$ and $Y$ cannot be distinguished by $k$-WL. The game is an Ehrenfeucht--Fra\"iss\'e game (c.f., \cite{Ebbinghaus:1994, Libkin}), with two players: Spoiler and Duplicator. We begin with $k+1$ pairs of pebbles, which are placed beside the graph. Each round proceeds as follows.
\begin{enumerate}
\item Spoiler picks up a pair of pebbles $(p_{i}, p_{i}^{\prime})$. 
\item We check the winning condition, which will be formalized below.
\item Duplicator chooses a bijection $f : V(X) \to V(Y)$.
\item Spoiler places $p_{i}$ on some vertex $v \in V(X)$. Then $p_{i}^{\prime}$ is placed on $f(v)$. 
\end{enumerate} 

Let $v_{1}, \ldots, v_{m}$ be the vertices of $X$ pebbled at the end of step 1, and let $v_{1}^{\prime}, \ldots, v_{m}^{\prime}$ be the corresponding pebbled vertices of $Y$. Spoiler wins precisely if the map $v_{\ell} \mapsto v_{\ell}^{\prime}$ does not extend to an isomorphism of the induced subgraphs $X[\{v_{1}, \ldots, v_{m}\}]$ and $Y[\{v_{1}^{\prime}, \ldots, v_{m}^{\prime}\}]$. Duplicator wins otherwise. Spoiler wins, by definition, at round $0$ if $X$ and $Y$ do not have the same number of vertices. $X$ and $Y$ are not distinguished by the first $r$ rounds of $k$-WL if and only if Duplicator wins the first $r$ rounds of the $(k+1)$-pebble game \cite{Hella1989, Hella1993, CFI}. 

Versions I and II of the pebble game are defined analogously, where Spoiler pebbles group elements. We first introduce the notion of marked equivalence. Let $\overline{u} := (u_{1}, \ldots, u_{k}), \overline{v} := (v_{1}, \ldots, v_{k})$ be $k$-tuples of group elements. We say that $\overline{u}$ and $\overline{v}$ are \textit{marked equivalent} in WL Version I iff (i) for all $i, j \in [k]$, $u_{i} = u_{j} \iff v_{i} = v_{j}$, and (ii) for all $i, j, \ell \in [k]$, $u_{i}u_{j} = u_{\ell} \iff v_{i}v_{j} = v_{\ell}$. We say that $\overline{u}$ and $\overline{v}$ are marked equivalent in WL Version II if the map $u_{i} \mapsto v_{i}$ extends to an isomorphism of the generated subgroups $\langle u_{1}, \ldots, u_{k} \rangle$ and $\langle v_{1}, \ldots, v_{k} \rangle$.

We now turn to formalizing Versions I and II of the pebble game. Precisely, for groups $G$ and $H$, each round proceeds as follows.
\begin{enumerate}
\item Spoiler picks up a pair of pebbles $(p_{i}, p_{i}^{\prime})$. 
\item We check the winning condition, which will be formalized below.
\item Duplicator chooses a bijection $f : G \to H$.
\item Spoiler places $p_{i}$ on some vertex $g \in G$. Then $p_{i}^{\prime}$ is placed on $f(g)$. 
\end{enumerate} 

Suppose that $(g_{1}, \ldots, g_{\ell}) \mapsto (h_{1}, \ldots, h_{\ell})$ have been pebbled. Duplicator wins at the given round if this map is a marked equivalence in the corresponding version of WL. Brachter \& Schweitzer established that for $J \in \{I, II\}$, $(k,r)$-WL Version J is equivalent to version J of the $(k+1)$-pebble, $r$-round pebble game \cite{WLGroups}.

\begin{remark}
In our work, we explicitly control for both pebbles and rounds. In our theorem statements, we state explicitly the number of pebbles on the board. So if Spoiler can win with $k$ pebbles on the board, then we are playing in the $(k+1)$-pebble game. Note that $k$-WL corresponds to $k$-pebbles on the board.
\end{remark}

Brachter \& Schweitzer \cite[Theorem~3.9]{WLGroups} also previously showed that WL Version I, II, and III are equivalent up to a factor of $2$ in the dimension, though they did not control for rounds. Following the proofs of Brachter \& Schweitzer \cite{WLGroups} for the bijective games, Grochow \& Levet \cite[Appendix~A]{GLWL1} showed that only $O(\log n)$ additional rounds are necessary.  

There exist analogous pebble games for count-free WL Versions I-III. The count-free $(k+1)$-pebble game consists of two players: Spoiler and Duplicator, as well as $(k+1)$ pebble pairs $(p, p^{\prime})$. In Versions I and II, Spoiler wishes to show that the two groups $G$ and $H$ are not isomorphic; and in Version III, Spoiler wishes to show that the corresponding graphs $\Gamma_{G}, \Gamma_{H}$ are not isomorphic. Duplicator wishes to show that the two groups (Versions I and II) or two graphs (Version III) are isomorphic. Each round of the game proceeds as follows.
\begin{enumerate}
\item Spoiler picks up a pebble pair $(p_{i}, p_{i}^{\prime})$.
\item The winning condition is checked. This will be formalized later.
\item In Versions I and II, Spoiler places one of the pebbles on some group element (either $p_{i}$ on some element of $G$ or $p_{i}'$ on some element of $H$). In Version III, Spoiler places one of the pebbles on some vertex of one of the graphs (either $p_{i}$ on some vertex of $\Gamma_{G}$ or $p_{i}'$ on some element of $\Gamma_{H}$).
\item Duplicator places the other pebble on some element of the other group (Versions I and II) or some vertex of the other graph (Version III).
\end{enumerate}

Let $v_{1}, \ldots, v_{m}$ be the pebbled elements of $G$ (resp., $\Gamma_{G}$) at the end of step 1, and let $v_{1}^{\prime}, \ldots, v_{m}^{\prime}$ be the corresponding pebbled vertices of $H$ (resp., $\Gamma_{H}$). Spoiler wins precisely if the map $v_{\ell} \mapsto v_{\ell}^{\prime}$ does not extend to a marked equivalence in the appropriate version of count-free WL. Duplicator wins otherwise. Spoiler wins, by definition, at round $0$ if $G$ and $H$ do not have the same number of elements. We note that $G$ and $H$ (resp., $\Gamma_{G}, \Gamma_{H}$) are not distinguished by the first $r$ rounds of the count-free $k$-WL if and only if Duplicator wins the first $r$ rounds of the count-free $(k+1)$-pebble game. Grochow \& Levet \cite{GLWL1} established the equivalence between Versions I and II of the count-free pebble game and the count-free WL algorithm for groups.

The count-free $r$-round, $k$-WL algorithm for graphs is equivalent to the $r$-round, $(k+1)$-pebble count-free pebble game \cite{CFI}. Thus, the count-free $r$-round, $k$-WL Version III algorithm for groups introduced in Brachter \& Schweitzer \cite{WLGroups} is equivalent to the $r$-round, $(k+1)$-pebble count-free pebble game on the graphs $\Gamma_G, \Gamma_H$ associated to the groups $G,H$.

\subsection{Logics} \label{sec:Logics}

We recall the central aspects of first-order logic. We have a countable set of variables $\{x_{1}, x_{2}, \ldots, \}$. Formulas are defined inductively. As our basis, $x_{i} = x_{j}$ is a formula for all pairs of variables. Now if $\varphi, \psi$ are formulas, then so are the following: $\varphi \land \psi, \varphi \vee \psi, \neg{\varphi}, \exists{x_{i}} \, \varphi,$ and $\forall{x_{i}} \, \varphi$. In order to define logics on groups, it is necessary to define a relation that relates the group multiplication. We recall the two different logics introduced by Brachter \& Schweitzer \cite{WLGroups}.
\begin{itemize}
\item \textbf{Version I:} We add a ternary relation $R$ where $R(x_{i}, x_{j}, x_{\ell}) = 1$ if and only if $x_{i}x_{j} = x_{\ell}$ in the group. In keeping with the conventions of \cite{CFI}, we refer to the first-order logic with relation $R$ as $\mathcal{L}^{I}$ and its $k$-variable fragment as $\mathcal{L}^{I}_{k}$. We refer to the logic $\mathcal{C}^{I}$ as the logic obtained by adding counting quantifiers $\exists^{\geq n} x_{i} \, \varphi$ (there exist at least $n$ distinct $x_{i}$ that satisfy $\varphi$) and $\exists{!n} \, \varphi$ (there exist exactly $n$ distinct $x_{i}$ that satisfy $\varphi$) and its $k$-variable fragment as $\mathcal{C}^{I}_{k}$. If furthermore we restrict the formulas to have quantifier depth at most $r$, we denote this fragment as $\mathcal{C}^{I}_{k,r}$.

\item \textbf{Version II:} We add a relation $R$, defined as follows. Let $w \in (\{x_{i_{1}}, \ldots, x_{i_{t}}\} \cup \{ x_{i_{1}}^{-1}, \ldots, x_{i_{t}}^{-1}\})^{*}$. We have that $R(x_{i_{1}}, \ldots, x_{i_{t}}; w) = 1$ if and only if multiplying the group elements according to $w$ yields the identity. For instance, $R(a, b; [a,b])$ holds precisely if $a, b$ commute. Again, in keeping with the conventions of \cite{CFI}, we refer to the first-order logic with relation $R$ as $\mathcal{L}^{II}$ and its $k$-variable fragment as $\mathcal{L}^{II}_{k}$. We refer to the logic $\mathcal{C}^{II}$ as the logic obtained by adding counting quantifiers $\exists^{\geq n} x_{i} \, \varphi$ and $\exists{!n} \, \varphi$ and its $k$-variable fragment as $\mathcal{C}^{II}_{k}$. If furthermore we restrict the formulas to have quantifier depth at most $r$, we denote this fragment as $\mathcal{C}^{II}_{k,r}$.
\end{itemize}

\begin{remark}
Brachter \& Schweitzer \cite{WLGroups} refer to the logics with counting quantifiers as $\mathcal{L}_{I}$ and $\mathcal{L}_{II}$. We instead adhere to the conventions in \cite{CFI}.
\end{remark}

\noindent Let $J \in \{ I, II\}$. Brachter \& Schweitzer \cite{WLGroups} established that two groups $G, H$ are distinguished by $(k,r)$-WL Version $J$ if and only if there exists a formula $\varphi \in \mathcal{C}^{J}_{k+1,r}$ such that $G \models \varphi$ and $H \not \models \varphi$. Following the techniques of Brachter \& Schweitzer, Grochow \& Levet \cite{GLWL1} established the analogous result for count-free WL Version $J$ and the logic $\mathcal{L}^{J}$. In the setting of graphs, the equivalence between Weisfeiler--Leman and first-order logic with counting quantifiers is well known \cite{ImmermanLander1990, CFI}.

\subsection{Weisfeiler--Leman as a Parallel Algorithm} \label{sec:ParallelWL}

Grohe \& Verbitsky \cite{GroheVerbitsky} previously showed that for fixed $k$, the classical $k$-dimensional Weisfeiler--Leman algorithm for graphs can be effectively parallelized. Precisely, each iteration of the classical counting WL algorithm (including the initial coloring) can be implemented using a logspace uniform $\textsf{TC}^{0}$ circuit, and each iteration of the \textit{count-free} WL algorithm can be implemented using a logspace uniform $\textsf{AC}^{0}$ circuit. As they mention (\cite[Remark~3.4]{GroheVerbitsky}), their implementation works for any first-order structure, including groups. However, because here we have three different versions of WL, we explicitly list out the resulting parallel complexities, which differ slightly between the versions.

\begin{itemize}
\item \textbf{WL Version I:} Let $(g_{1}, \ldots, g_{k})$ and $(h_{1}, \ldots, h_{k})$ be two $k$-tuples of group elements. We may test in $\textsf{AC}^{0}$ whether (a) for all $i, j, m \in [k]$, $g_{i}g_{j} = g_{m} \iff h_{i}h_{j} = h_{m}$, and (b) $g_{i} = g_{j} \iff h_{i} = h_{j}$. So we may decide if two $k$-tuples receive the same initial color in $\textsf{AC}^{0}$. Comparing the multiset of colors at the end of each iteration (including after the initial coloring), as well as the refinement steps, proceed identically as in \cite{GroheVerbitsky}. Thus, for fixed $k$, each iteration of $k$-WL Version I can be implemented using a logspace uniform $\textsf{TC}^{0}$ circuit. In the setting of the count-free $k$-WL Version I, we are comparing the set rather than multiset of colors at each iteration. So each iteration (including the initial coloring) can be implemented using a logspace uniform $\textsf{AC}^{0}$ circuit.

\item \textbf{WL Version II:} Let $(g_{1}, \ldots, g_{k})$ and $(h_{1}, \ldots, h_{k})$ be two $k$-tuples of group elements. We may use the marked isomorphism test of Tang \cite{TangThesis} to test in $\textsf{L}$ whether the map sending $g_{i} \mapsto h_{i}$ for all $i \in [k]$ extends to an isomorphism of the generated subgroups $\langle g_{1}, \ldots, g_{k} \rangle$ and $\langle h_{1}, \ldots, h_{k} \rangle$. So we may decide whether two $k$-tuples receive the same initial color in $\textsf{L}$. Comparing the multiset of colors at the end of each iteration (including after the initial coloring), as well as the refinement steps, proceed identically as in \cite{GroheVerbitsky}. Thus, for fixed $k$, the initial coloring of $k$-WL Version II is $\textsf{L}$-computable, and each refinement step is $\textsf{TC}^{0}$-computable. In the case of the count-free $k$-WL Version II, the initial coloring is still $\textsf{L}$-computable, while each refinement step can be implemented can be implemented using a logspace uniform $\textsf{AC}^{0}$ circuit.
\end{itemize}

\subsection{Further Related Work} \label{sec:RelatedWork}

In addition to the intrinsic interest of this natural problem, a key motivation for the $\algprobm{Group Isomorphism}$ problem (\algprobm{GpI}) is its close relation to the \algprobm{Graph Isomorphism} problem ($\algprobm{GI}$). In the Cayley (verbose) model, $\algprobm{GpI}$ reduces to $\algprobm{GI}$ \cite{ZKT}, while $\algprobm{GI}$ reduces to the succinct $\algprobm{GpI}$ problem \cite{Heineken1974TheOO, Mekler} (recently simplified \cite{HeQiao}). In light of Babai's breakthrough result that $\algprobm{GI}$ is quasipolynomial-time solvable \cite{BabaiGraphIso}, $\algprobm{GpI}$ in the Cayley model is a key barrier to improving the complexity of $\algprobm{GI}$. Both verbose $\algprobm{GpI}$ and $\algprobm{GI}$ are considered to be candidate $\textsf{NP}$-intermediate problems, that is, problems that belong to $\textsf{NP}$, but are neither in $\textsf{P}$ nor $\textsf{NP}$-complete \cite{Ladner}. There is considerable evidence suggesting that $\algprobm{GI}$ is not $\textsf{NP}$-complete \cite{Schoning, BuhrmanHomer, ETH, BabaiGraphIso, GILowPP, ArvindKurur, MATHON1979131}. As verbose $\algprobm{GpI}$ reduces to $\algprobm{GI}$, this evidence also suggests that $\algprobm{GpI}$ is not $\textsf{NP}$-complete. It is also known that $\algprobm{GI}$ is strictly harder than $\algprobm{GpI}$ under $\textsf{AC}^{0}$ reductions \cite{ChattopadhyayToranWagner}. Tor\'an showed that $\algprobm{GI}$ is $\textsf{DET}$-hard \cite{Toran}. In particular, $\algprobm{Parity}$ is $\textsf{AC}^{0}$-reducible to $\algprobm{GI}$ \cite{CFI, Toran}. On the other hand, Chattopadhyay, Tor\'an, and Wagner showed that $\algprobm{Parity}$ is not $\textsf{AC}^{0}$-reducible to $\algprobm{GpI}$ \cite{ChattopadhyayToranWagner}. To the best of our knowledge, there is no literature on lower bounds for $\algprobm{GpI}$ in the Cayley table model. The best known complexity-theoretic upper bound for \algprobm{GI} is $\mathbb{F}$-\algprobm{Tensor Isomorphism} ($\textsf{TI}_{\mathbb{F}}$). When the field $\mathbb{F}$ is finite, $\textsf{TI}_{\mathbb{F}} \subseteq \textsf{NP} \cap \textsf{coAM}$ \cite{GrochowQiaoTensors}. Thus, in a precise sense, \algprobm{Graph Isomorphism} is at least as hard as linear algebra and is no harder than multilinear algebra.

Combinatorial techniques, such as individualization and refinement, have also been incredibly successful in $\algprobm{GI}$, yielding efficient isomorphism tests for several families \cite{GroheVerbitsky, KieferPonomarenkoSchweitzer, GroheKieferPlanar, grohe_et_al:LIPIcs:2019:10693, grohe2019canonisation, BabaiWilmes, ChenSunTeng}.  Weisfeiler--Leman is also a key subroutine in Babai's quasipolynomial-time isomorphism test \cite{BabaiGraphIso}. Despite the successes of such combinatorial techniques, they are known to be insufficient to place $\algprobm{GI}$ into $\textsf{P}$ \cite{CFI, NeuenSchweitzerIR}. In contrast, the use of combinatorial techinques for $\algprobm{GpI}$ is relatively new \cite{QiaoLiWL, BGLQW, WLGroups, BrachterSchweitzerWLLibrary, GLWL1, GLDescriptiveComplexity}, and with the key open question in the area being  whether such techniques are sufficient to improve even the long-standing upper-bound of $n^{\Theta(\log n)}$ runtime.

Even in the setting of graphs, there is little known about count-free Weisfeiler--Leman. It is known that count-free WL fails to place $\algprobm{GI}$ into $\textsf{P}$. In particular, almost all pairs of graphs are indistinguishable by the constant-dimensional count-free WL \cite{IMMERMAN198276, FaginCountFree}. Nonetheless, Grohe \& Verbitsky \cite{GroheVerbitsky} used the count-free WL algorithm to distinguish rotation systems in $O(\log n)$ rounds. As a consequence, they obtained a new proof that isomorphism testing of planar graphs was in $\textsf{AC}^{1}$. Verbitsky subsequently showed that the $14$-dimensional count-free WL algorithm places isomorphism testing of $3$-connected planar graphs into $\textsf{AC}^{1}$, removing the reduction for planar graphs to rotation systems in this special case \cite{VerbitskyPlanarCountFree}. Grochow \& Levet \cite{GLWL1} showed that in the setting of groups, count-free WL is unable to even distinguish Abelian groups in polynomial-time. Nonetheless, after $O(\log \log n)$ rounds, the multiset of colors computed by count-free WL will differ for an Abelian group $G$ and a non-isomorphic (not necessarily Abelian) group $H$. 

There has been considerable work on efficient parallel ($\textsf{NC}$) isomorphism tests for graphs \cite{LindellTreeCanonization, BirgitKoblerMcKenzieToran, KoblerVerbitsky,WagnerBoundedTreewidth, ElberfeldSchweitzer, GroheVerbitsky, GroheKieferPlanar, DattaLimayeNimbhorkarPrajaktaThieraufWagner, datta_et_al:LIPIcs:2009:2314, ARVIND20121}. In contrast with the work on serial runtime complexity, the literature on the space and parallel complexity for $\algprobm{GpI}$ is quite minimal. Around the same time as Tarjan's $n^{\log_{p}(n) + O(1)}$-time algorithm for $\algprobm{GpI}$ \cite{MillerTarjan}, Lipton, Snyder, and Zalcstein showed that $\algprobm{GpI} \in \textsf{SPACE}(\log^{2}(n))$ \cite{LiptonSnyderZalcstein}. This bound has been improved to to $\beta_{2}\textsf{AC}^{1}$ ($\textsf{AC}^{1}$ circuits that receive $O(\log^{2}(n))$ non-deterministic bits as input)\footnote{Wolf claimed a bound of $\beta_{2}\textsf{NC}^{2}$. However, he used $\textsf{NC}^{1}$ circuits to multiple two group elements, whereas an $\textsf{AC}^{0}$ circuit suffices.} \cite{Wolf}, and subsequently to $\beta_{2}\textsf{L} \cap \beta_{2}\textsf{FOLL} \cap \beta_{2}\textsf{SC}^{2}$ \cite{ChattopadhyayToranWagner, TangThesis}. In the case of Abelian groups, Chattopadhyay, Tor\'an, and Wagner showed that $\algprobm{GpI} \in \textsf{L} \cap \textsf{TC}^{0}(\textsf{FOLL})$ \cite{ChattopadhyayToranWagner}. Grochow \& Levet \cite{GLWL1} subsequently observed that the algorithm of \cite{ChattopadhyayToranWagner} could be implemented in $\beta_{1}\textsf{MAC}^{0}(\textsf{FOLL})$, further improving the complexity of identifying Abelian groups. Tang showed that isomorphism testing for groups with a bounded number of generators can also be done in $\textsf{L}$ \cite{TangThesis}. Since composition factors of permutation groups can be identified in $\textsf{NC}$ \cite{BabaiLuksSeress} (see also \cite{BealsCompositionFactors} for a CFSG-free proof), isomorphism testing \emph{between} two groups that are both direct products of simple groups (Abelian or non-Abelian) can be done in $\textsf{NC}$, using the regular representation, though this does not allow one to test isomorphism of such a group against an arbitrary group. Grochow \& Levet \cite{GLWL1} showed that Weisfeiler--Leman can compute the direct factors of a group in $\mathsf{NC}$, improving on the serial bound of Brachter \& Schweitzer \cite{BrachterSchweitzerWLLibrary}. Furthermore, Grochow \& Levet \cite{GLWL1} showed that WL serves as an $\textsf{L}$-computable isomorphism test for coprime extensions $H \ltimes N$ where $H$ is $O(1)$-generated and $N$ is Abelian.

Examining the distinguishing power of the counting logic $\mathcal{C}_{k}$ serves as a measure of descriptive  complexity for groups. In the setting of graphs, the descriptive complexity has been extensively studied, with \cite{GroheBook} serving as a key reference in this area. There has been recent work relating first order logics and groups \cite{FiniteGroupsFOL, WLGroups, BrachterSchweitzerWLLibrary}, as well as work examining the descriptive complexity of finite abelian groups \cite{DescriptiveComplexityAbelianGroups}. 

Ehrenfeucht--Fra\"iss\'e games \cite{Ehrenfeucht, Fraisse}, also known as pebbling games, serve as another tool in proving the inexpressibility of certain properties in first-order logics. 
Two finite structures are said to be elementary equivalent if they satisfy the same first-order sentences. In such games, we have two players who analyze two given structures. Spoiler seeks to prove the structures are not elementary equivalent, while Duplicator seeks to show that the structures are indeed elementary equivalent. Spoiler begins by selecting an element from one structure, and Duplicator responds by picking a similar element from the other structure. Spoiler wins if and only if the eventual substructures are not isomorphic. 
Pebbling games have served as an important tool in analyzing graph properties like reachability \cite{ajtai_fagin_1990, ARORA199797}, designing parallel algorithms for graph isomorphism \cite{GroheVerbitsky}, and isomorphism testing of random graphs \cite{Rossman2009EhrenfeuchtFrassGO}.

\section{Bounded-Generator Solvable Groups}

In this section, show that the count-free WL Version I serves as an $\textsf{FO}(\poly \log \log n)$-canonization procedure for $O(1)$-generated solvable groups with solvability class $\poly \log \log n$. We begin with the following theorem.

\begin{theorem} \label{thm:BoundedGenSolvable}
Let $G$ be an $O(1)$-generated solvable group with solvability class $r$, and let $H$ be arbitrary. The count-free $(O(1), O(r \log \log n))$-WL Version I algorithm can distinguish $G$ and $H$.
\end{theorem}

\begin{proof}
Let $d := d(G)$. Suppose that generators $(g_{1}, \ldots, g_{d})$ for $G$ have been individualized. By \cite[Lemma~7.12]{GLWL1}, powers of individualized elements will receive a unique color after $O(\log \log n)$ rounds of count-free WL Version III. By \cite[Lemma~7.9]{GLWL1}, the result also holds for count-free WL Version I. Now by \cite[Lemma~7.3]{GLWL1} and \cite[Lemma~7.9]{GLWL1}, words of length $O(\log n)$ over individualized elements will receive a unique color after $O(\log \log n)$ iterations of count-free WL Version I. It follows by \cite[Theorem~3.5]{BKLM} that if $G$ has solvability class $r$, that $O(r \log \log n)$ rounds of count-free WL Version I will assign a unique color to each element of $G$. The result now follows.
\end{proof}

We immediately obtain the following corollary.

\begin{corollary} \label{cor:SolvableIsomorphism}
Let $G$ be an $O(1)$-generated solvable group with solvability class $(\log \log n)^{c}$, and let $H$ be arbitrary. We can decide isomorphism between $G$ and $H$ in $\textsf{FO}((\log \log n)^{c+1})$.
\end{corollary}

\begin{remark} \label{rmk:SolvableCanonization}
It is possible to upgrade \Thm{thm:BoundedGenSolvable} to obtain $\textsf{quasiFO}((\log \log n)^{c+1})$ canonization. Let $G$ be a $d$-generated solvable group with solvability class $(\log \log n)^{c}$. We first run the count-free $(d, O((\log \log n)^{c+1}))$-WL Version I. For each color class containing an element (and thus, all elements) of the form $(g_{1}, \ldots, g_{d})$ where $g_{1}, \ldots, g_{d}$ are all distinct, we individualize $g_{1}, \ldots, g_{d}$ and run the count-free $(d, O((\log \log n)^{c+1}))$-WL Version I. To obtain a canonical form, we take a representative $(g_{1}, \ldots, g_{d})$ such that each color class has size $1$ after individualizing $(g_{1}, \ldots, g_{d})$ and running the count-free $(d, O((\log \log n)^{c+1}))$-WL Version I. 

To see that finding the minimum color class (under the natural ordering of the labels of the color classes) is $\textsf{AC}^{0}$-computable given the coloring, we appeal to the characterization that $\textsf{AC}^{0} = \textsf{FO}$ \cite{MIXBARRINGTON1990274}. We may write down a first-order formula for the minimum element, and so finding the minimum color class is $\textsf{AC}^{0}$-computable given the coloring. Furthermore, identifying the members of a given color class is $\textsf{AC}^{0}$-computable. Thus, in total, we only require an $\textsf{quasiFOLL}$ circuit for canonization. If $d \in O(1)$, we have an $\textsf{FO}((\log \log n)^{c+1})$ circuit. We record this with the following theorem.
\end{remark}

\begin{theorem}
Canonical forms for solvable groups with solvability class $O((\log \log n)^{c})$ can be computed using a $\textsf{quasiFO}((\log \log n)^{c+1})$ circuit. Furthermore, for $O(1)$-generated solvable groups with solvability class $O((\log \log n)^{c})$, canonical forms can be computed using an $\textsf{FO}((\log \log n)^{c+1})$ circuit.
\end{theorem}

\begin{remark}
As nilpotent groups have solvability class $O(\log \log n)$ \cite{Therien1980ClassificationOR}, we obtain the following consequence.
\end{remark}

\begin{corollary}
Canonical forms for nilpotent groups can be computed using a $\textsf{quasiFO}((\log \log n)^{2})$ circuit. Furthermore, for $O(1)$-generated nilpotent groups, canonical forms can be computed using an $\textsf{FO}((\log \log n)^{2})$ circuit.
\end{corollary}

We now turn our attention to canonizing finite simple groups. Grochow \& Levet \cite[Corollary~7.14]{GLWL1} previously observed that finite simple groups are identified by the $(O(1), O(\log \log n))$-WL Versions I and III. They appealed to the fact that finite simple groups are uniquely identified amongst all groups by their order and the set of orders of their elements \cite{SimpleOrder}. While this result suffices for an $\textsf{FOLL}$ isomorphism test, it is not clear how to obtain $\textsf{FOLL}$ canonization for finite simple groups. To this end, we recall the following result of Babai, Kantor, \& Lubotzky \cite{BabaiKantorLubotsky}.

\begin{theorem}[{\cite{BabaiKantorLubotsky}}] \label{thm:BabaiKantorLubotzky}
There exists an absolute constant $C$ such that for every non-Abelian finite simple group $S$, there exists a generating sequence $(s_{1}, \ldots, s_{7})$ such that every element of $S$ can be realized as a word of length $C \cdot \log |S|$ over $s_{1}, \ldots, s_{7}$.
\end{theorem}

This yields an alternative proof that finite simple groups admit an $\textsf{FOLL}$ isomorphism test.

\begin{corollary} \label{cor:FOLLSimple}
Let $G$ be a finite simple group, and let $H$ be arbitrary. We can decide isomorphism between $G$ and $H$ in $\textsf{FOLL}$.
\end{corollary}

To obtain $\textsf{FOLL}$ canonization, we proceed identically as in \Rmk{rmk:SolvableCanonization}:

\begin{theorem}
Let $G$ be a finite simple group. We can compute a canonical labeling for $G$ in $\textsf{FOLL}$.
\end{theorem}

\section{CFI Groups}

\subsection{Preliminaries}

\textbf{CFI Construction.} Cai, F\"urer, \& Immerman \cite{CFI} previously established that for every $k$, there exist a pair of graphs that are indistinguishable by $k$-WL. We recall their construction, which we denote the CFI construction, here. We begin with a connected base graph $\Gamma$. In $\Gamma$, each vertex is replaced by a particular gadget, and the gadgets are interconnected according to the edges of $\Gamma$ as follows. For a vertex of degree $d$, we define the gadget $F_{d}$ to be the graph whose vertex set consists of a set of external vertices $O_{d} = \{ a_{1}^{v}, b_{1}^{v}, \ldots, a_{d}^{v}, b_{d}^{v}\}$ and a set of internal vertices $M_{d}$ which are labeled according to the strings in $\{0,1\}^{d}$ that have an even number of $1$'s. For each $i$, each internal vertex $u$ of $M_{d}$ is adjacent to exactly one of $\{ a_{i}^{v}, b_{i}^{v}\}$; namely $u$ is adjacent to $a_{i}$ if the $i$th bit of the string is $0$ and $b_{i}$ otherwise. An example of $F_{3}$ is included here (see Figure \ref{CFIFigure}).

\begin{figure} 
\begin{center}
\includegraphics{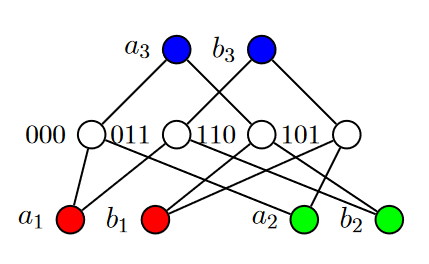}
\end{center}
\caption{The CFI gadget $F_{3}$ \cite{CFI, WLGroups}.}
\label{CFIFigure}
\end{figure}

We now discuss how the gadgets are interconnected. Let $xy \in E(\Gamma)$. For each pair of external vertices $(a_{i}^{x}, b_{i}^{x})$ on the gadget corresponding to $x$ and each pair of external vertices $(a_{j}^{y}, b_{j}^{y})$ on the gadget corresponding to $y$, we add parallel edges $a_{i}^{x}a_{j}^{y}, b_{i}^{x}b_{j}^{y}$. We refer to the resulting graph as $\text{CFI}(\Gamma)$. The \textit{twisted CFI-graph} $\widetilde{\text{CFI}(\Gamma)}$ is obtained by taking one pair of parallel edges $a_{i}^{x}a_{j}^{y}, b_{i}^{x}b_{j}^{y}$ from $\text{CFI}(\Gamma)$ and replacing these edges with the twist $a_{i}^{x}b_{j}^{y}, b_{i}^{x}a_{j}^{y}$. Up to isomorphism, it does not matter which pair of parallel edges we twist \cite{CFI}. For a subset of edges $E' \subseteq E(\Gamma)$ of the base graph, we can define the graph obtained by twisting exactly the edges in $E'$. The resulting graph is isomorphic to $\text{CFI}(\Gamma)$ if $|E'|$ is even and isomorphic to $\widetilde{\text{CFI}(\Gamma)}$ otherwise.

In the original construction \cite{CFI}, the base graph is generally taken to be a vertex-colored graph where each vertex has a different color. As a result, all of the gadgets in the CFI construction are distinguishable. The colors can be removed by instead attaching gadgets to each vertex, and these gadgets can be attached in such that the base graph is identified by $2$-WL. In particular, it is possible to choose a $3$-regular base graph to have WL-dimension $2$ \cite[Observation~2.2]{WLGroups}.

\noindent  \\ \textbf{Mekler's Construction.} We recall Mekler's construction \cite{Mekler} (recently improved by He \& Qiao \cite{HeQiao}), which allows us to encode a graph into a class $2$ $p$-group ($p > 2$) of exponent $p$.

\begin{definition}
For $n \in \mathbb{N}$ and a prime $p > 2$, the relatively free class $2$ $p$-group of exponent $p$ is given by the presentation
\[
F_{n,p} = \langle x_{1}, \ldots, x_{n} | R(p,n)\rangle,
\]

\noindent where $R(p, n)$ consists of the following relations:
\begin{itemize}
\item For all $i \in [n]$, $x_{i}^{p} = 1$, and 
\item For all $i, j, k \in [n]$, $[[x_{i}, x_{j}], x_{k}] = 1$.
\end{itemize}

\noindent \\ Thus, $F_{n,p}$ is generated by $x_{1}, \ldots, x_{n}$, each of these generators has order $p$, and the commutator of any two generators commutes with every generator (and thus, every group element). It follows that each element of $F_{n,p}$ can be written uniquely in the following normal form:
\[
x_{1}^{d_{1}} \cdots x_{n}^{d_{n}} [x_{1}, x_{2}]^{d_{1,2}} \cdots [x_{1}, x_{n}]^{d_{1,n}} [x_{2}, x_{3}]^{d_{2,3}} \cdots [x_{n-1}, x_{n}]^{d_{n-1},n}.
\]

\noindent Here the exponents take on values in $\{0, \ldots, p-1\}$. In particular, $|F_{n,p}| = p^{n+ n(n-1)/2}.$
\end{definition}

Mekler's construction \cite{Mekler, HeQiao} allows us to encode a graph as a class $2$ $p$-group of exponent $p$ as follows.

\begin{definition}[Mekler's Construction]
Let $\Gamma(V, E)$ be a simple, undirected graph with $V = \{ v_{1}, \ldots, v_{n}\}$, and let $p > 2$ be prime. We construct a class $2$ $p$-group of exponent $p$ as follows:
\[
G_{\Gamma} = \langle x_{1}, \ldots, x_{n} \, \bigr| \, R(p,n), [x_{i}, x_{j}] = 1 : v_{i}v_{j} \in E \rangle.
\]

\noindent So two generators of $G_{\Gamma}$ commute precisely if the corresponding vertices form an edge of $\Gamma$. We identify $x_{i}$ with the vertex $v_{i}$.
\end{definition}

\begin{remark}
Mekler's construction provides a many-one reduction from $\algprobm{GI}$ to $\algprobm{GpI}$ (or in the language of category theory, an isomorphism-preserving functor from the category of graphs to the category of groups). In his original construction, Mekler first reduced arbitrary graphs to ``nice" graphs via the use of gadgets. He \& Qiao  \cite{HeQiao} subsequently showed that this gadgetry was unnecessary. In both Mekler's original construction \cite{Mekler} and the improvement due to He \& Qiao \cite{HeQiao}, this reduction is polynomial-time computable when the inputs for the groups are given in the generator-relation model (c.f., \cite{Mekler, HeQiao} and \cite[Theorem~4.13]{WLGroups}). While we will be given such groups by their Cayley (multiplication) tables, we are still able to reason about the groups using the underlying graph-theoretic structure.
\end{remark}

We now recall some key properties about the groups arising via Mekler's construction.

\begin{lemma}[{\cite[Lemma~4.3]{WLGroups}}]
We have that $\Phi(G_{\Gamma}) = G_{\Gamma}'$, and the vertices of $\Gamma$ form a minimum-cardinality generating set of $G_{\Gamma}$.
\end{lemma}

\begin{lemma}[{\cite[Corollary~4.5]{WLGroups}}]
Let $\Gamma$ be a simple, undirected graph. Then we have that $|G_{\Gamma}| = p^{|V(\Gamma)| + |\binom{V}{2} - E(\Gamma)|}$. In particular, every element of $G_{\Gamma}$ can be written in the form:
\[
v_{1}^{d_{1}} \cdots v_{n}^{d_{n}} c_{1}^{d_{n+1}} \cdots c_{k}^{d_{n+k}},
\]

\noindent where $\{ c_{1}, \ldots, c_{k}\}$ is the set of non-trivial commutators between generators (i.e., the non-edges of $\Gamma$) and each $d_{i}$ is uniquely determined modulo $p$.
\end{lemma}

\begin{lemma}[{\cite[Lemma~4.7]{WLGroups}}] \label{lem:BSFrattini}
Let $\Gamma$ be a simple, undirected graph. We have $Z(G_{\Gamma}) = G_{\Gamma}' \times \langle v : N[v] = V(\Gamma)\rangle$. In particular, if no vertex of $\Gamma$ is adjacent to every other vertex, then $Z(G_{\Gamma}) = G_{\Gamma}'$.
\end{lemma}

\begin{definition}[{\cite[Definition~4.8]{WLGroups}}]
Let $x \in G_{\Gamma}$ be an element with normal form:
\[
x := v_{1}^{d_{1}} \cdots v_{n}^{d_{n}} c_{1}^{e_{1}} \cdots c_{m}^{e_{m}}.
\]

\noindent The \textit{support} of $x$ is $\{ v_{i} : d_{i} \not \equiv 0 \pmod{p} \}$. For a subset $S = \{ v_{i_{1}}, \ldots, v_{i_{s}}\} \subseteq V(\Gamma)$, let $x_{S}$ be the subword $v_{i_{1}}^{d_{i_{1}}} \cdots v_{i_{s}}^{d_{i_{s}}}$, with $i_{1} < \cdots < i_{s}$.
\end{definition}

\begin{lemma}[{\cite[Corollary~4.11]{WLGroups}}] \label{lem:BSCentralizers}
Let $\Gamma$ be a simple, undirected graph. Let $x := v_{i_{1}}^{d_{1}} \cdots v_{i_{r}}^{d_{r}}c$, with $i_{1} < i_{2} < \cdots < i_{r}$, $c$ central in $G_{\Gamma}$, and each $d_{i} \not \equiv 0 \pmod{p}$. Let $C_{1}, \ldots, C_{s}$ be the connected components of the complement graph $\text{co}(\Gamma[\supp(x)])$. Then:
\[
C_{G_{\Gamma}}(x) = \langle x_{C_{1}} \rangle \cdots \langle x_{C_{s}} \rangle \langle \{ v_{m} : [v_{m}, v_{i_{j}}] = 1 \text{ for all } j \} \rangle G_{\Gamma}'.
\]
 \end{lemma}

\subsection{Weisfeiler--Leman and the CFI Groups}

In this section, we establish the following.

\begin{theorem} \label{thm:CFIWL}
Let $\Gamma_{0}$ be a $3$-regular connected graph, and let $\Gamma_{1} := \text{CFI}(\Gamma_{0})$ and $\Gamma_{2} := \widetilde{\text{CFI}(\Gamma_{0})}$ be the corresponding CFI graphs. For $i = 1, 2$, denote $G_{i} := G_{\Gamma_{i}}$ be the corresponding groups arising from Mekler's construction. We have that the $(3, O(\log  \log n))$-WL Version I distinguishes $G_{1}$ from $G_{2}$. If furthermore $\Gamma_{0}$ is identified by the graph $(3, r)$-WL algorithm, then the $(3, \max\{r, O(\log \log n)\})$-WL Version I algorithm identifies $G_{1}$ as well as $G_{2}$.
\end{theorem}

\noindent
\begin{remark}
Brachter \& Schweitzer \cite[Theorem~6.1]{WLGroups} previously showed that $3$-WL Version II sufficed to identify the CFI groups. In closely analyzing their work, we see that their proof shows that only $O(\log n)$ rounds suffice. In light of the parallel WL implementation for graphs due to Grohe \& Verbitsky \cite{GroheVerbitsky}, this places isomorphism testing of $G_{1}$ and $G_{2}$ into $\textsf{TC}^{1}$. If furthermore, the base graph is identified by the graph $3$-WL algorithm in $O(\log n)$ rounds, then we can decide whether $G_{i}$ ($i = 1, 2$) and an arbitrary group $H$ are isomorphic in $\textsf{TC}^{1}$.

WL Version I and WL Version II admit different parallel complexities, due to their initial colorings. Namely, the initial coloring of WL Version I is $\textsf{TC}^{0}$-computable, while the best known bound for the initial coloring of WL Version II is $\textsf{L}$ due to Tang's marked isomorphism test \cite{TangThesis}. The refinement step for both WL Versions I and II is $\textsf{TC}^{0}$-computable. So while the number of rounds can be improved, the initial coloring is still a barrier for improving the parallel complexity of WL Version II. However, the analysis of the CFI groups in \cite{WLGroups} relies almost exclusively on the underlying graph-theoretic structure of the CFI graphs. That is, we do not need the full power of deciding whether two $3$-tuples of group elements generate isomorphic subgroups. Thus, we may use the weaker Weisfeiler--Leman Version I algorithm, whose initial coloring is $\textsf{TC}^{0}$-computable. We obtain the following corollary. 
\end{remark}

\begin{corollary}
\noindent Using the same notation as in \Thm{thm:CFIWL}, we have the following.
\begin{enumerate}[label=(\alph*)]
\item We may decide whether $G_{1} \cong G_{2}$ using a logspace uniform $\textsf{TC}$ circuit of depth $O(\log \log n)$ (which is contained in $\textsf{TC}^{o(1)}$). 
\item If furthermore, the base graph $\Gamma_{0}$ is identifiable by the graph $3$-WL in $r$ rounds, then we can decide isomorphism between $G_{i}$ ($i = 1, 2$) and an arbitrary group $H$ using a logspace uniform $\textsf{TC}$ circuit of depth $\max\{r, O(\log \log n)\}$.
\end{enumerate}
\end{corollary}

\noindent \\ We also improve the descriptive complexity of the CFI groups, by showing that the quantifier depth can be reduced from $O(\log n)$ to $O(\log \log n)$. Furthermore, we obtain this improvement in the weaker Version I logic, whereas Brachter \& Schweitzer \cite{WLGroups} utilized the stronger Version II logic.

\begin{corollary}
\noindent Let $\mathcal{C}^{I}_{k,r}$ be the Version I fragment of first-order logic with counting quantifiers, where formulas are permitted at most $k$ variables and quantifier depth at most $r$ (see Section~\ref{sec:Logics}). Using the same notation as in \Thm{thm:CFIWL}, we have the following.
\begin{enumerate}[label=(\alph*)]
\item There exists a formula $\varphi$ in $\mathcal{C}^{I}_{3, O(\log \log n)}$ such that $G_{1} \models \varphi$ and $G_{2} \not \models \varphi$, or vice-versa.

\item If furthermore, the base graph $\Gamma_{0}$ is identifiable by the graph $3$-WL in $r$ rounds, then for any group $H \not \cong G$, there is a formula $\varphi_{i}$ in $\mathcal{C}^{I}_{3, \max\{r, O(\log \log n)\}}$ such that $G_{i} \models \varphi_{i}$ and $H \not \models \varphi_{i}$ (or vice-versa).
\end{enumerate}
\end{corollary}

\noindent \\ We begin with some preliminary lemmas.

\begin{lemma} \label{lem:Centralizer}
Let $G, H$ be groups. Suppose Duplicator selects a bijection $f : G \to H$ such that $|C_{G}(x)| \neq |C_{H}(f(x))|$. Then Duplicator can win with $3$ pebbles and $3$ rounds.
\end{lemma}

\begin{proof}
Without loss of generality, suppose that $|C_{G}(x)| > |C_{H}(f(x))|$. Spoiler begins by pebbling $x \mapsto f(x)$. Let $f' : G \to H$ be the bijection that Duplicator selects the next round. As $|C_{G}(x)| > |C_{H}(f(x))|$, there exists $y \in C_{G}(x)$ such that $f'(y) \not \in C_{H}(f(x))$. At the next two rounds, Spoiler pebbles $y, xy$ and wins.
\end{proof}

\begin{remark}
Brachter \& Schweitzer \cite{WLGroups} established that for the CFI groups $G_{1}, G_{2}$,  group elements with single-vertex support have centralizers of size $p^{4} \cdot |Z(G)|$, while all other group elements have centralizers of size at most $p^{3} \cdot |Z(G)|$. So by \Lem{lem:Centralizer}, if Duplicator does not preserve single-support vertices, then Spoiler can win with $2$ additional pebbles and $2$ additional rounds.
\end{remark}

\begin{lemma}[Compare rounds c.f. {\cite[Lemma~6.3]{WLGroups}}] \label{lem:Support}
Let $G_{1}, G_{2}$ be the CFI groups. Let $f : G_{1} \to G_{2}$ be the bijection that Duplicator selects. If there exists $x \in G_{1}$ such that $|\supp(x)| \neq |\supp(f(x))|$, then Spoiler can win with $3$ pebbles and $O(\log \log |G_{1}|)$ rounds.
\end{lemma}

\begin{proof}
Without loss of generality, suppose that $|\supp(x)| < |\supp(f(x))|$. Spoiler begins by pebbling $x \mapsto f(x)$. Write $x = x_{i_{1}}^{d_{1}} \cdots x_{i_{r}}^{d_{r}} \cdot c$, where the $x_{i_{1}}, \ldots, x_{i_{r}}$ correspond to vertices $v_{i_{1}}, \ldots, v_{i_{r}} \in V(\Gamma_{1})$ as per Mekler's construction, and $c \in Z(G_{1})$. Let $m := \lceil r/2 \rceil$. At the next two rounds, Spoiler pebbles $x_{i_{1}}^{d_{1}} \cdots x_{i_{m}}^{d_{m}} \mapsto u$ and $x_{i_{m+1}}^{d_{i_{m+1}}} \cdots x_{i_{r}}^{d_{r}} \cdot c \mapsto v$. Now if $f(x) \neq uv$, then Spoiler wins. So suppose $f(x) = uv$. As $|\supp(x)| \neq |\supp(f(x))|$, we have that either:
\begin{align*}
&|\supp(x_{i_{1}}^{d_{1}} \cdots x_{i_{m}}^{d_{m}})| < |\supp(u)| \text{ or, } \\ 
&|\supp(x_{i_{m+1}}^{d_{i_{m+1}}} \cdots x_{i_{r}}^{d_{r}} \cdot c)| < |\supp(v)|.
\end{align*}

\noindent Without loss of generality, suppose that: $|\supp(x_{i_{1}}^{d_{1}} \cdots x_{i_{m}}^{d_{m}})| < |\supp(u)|$. Spoiler iterates on the above argument starting from $x_{i_{1}}^{d_{1}} \cdots x_{i_{m}}^{d_{m}} \mapsto u$ and reusing the pebbles on $x, x_{i_{m+1}}^{d_{i_{m+1}}} \cdots x_{i_{r}}^{d_{r}} \cdot c$. Now as $|\supp(x_{i_{1}}^{d_{1}} \cdots x_{i_{m}}^{d_{m}})| \leq |\supp(x)/2|$, we iterate on this argument at most $\log_{2}(|\supp(x)|) + 1 \leq \log |V(\Gamma_{1})| \leq \log \log |G_{1}| + O(1)$ times until we reach a base case where our pebbled element $x' \in G_{1}$ has support size $1$, but the corresponding pebbled element $y' \in H$ has support size $> 1$. In this case, Spoiler at the next round reuses one of the other two pebbles on the board to pebble some element in $C_{G_{1}}(x')$ whose image does not belong to $C_{G_{2}}(y')$. The result now follows.
\end{proof}

\begin{lemma}[Compare rounds c.f. {\cite[Lemma~6.4]{WLGroups}}] \label{lem:GadgetStructure}
Let $u \in V(\Gamma_{1})$, and let $g_{u} \in G_{1}$ be a single-support element that is supported by $u$. 
\begin{enumerate}[label=(\alph*)]
\item Suppose $v \in V(\Gamma_{1})$ belongs to the same gadget as $u$, and let $g_{v} \in G_{1}$ be a single-support element that is supported by $v$. Let $f : G_{1} \to G_{2}$ be the bijection that Duplicator selects. If $f(g_{u}g_{v})$ is not supported by exactly two vertices $x, y$ on the same gadget of $V(\Gamma_{2})$, then Spoiler can win with $3$ pebbles and $O(1)$ rounds.

\item Suppose that $g_{u} \mapsto x$ has been pebbled. Let $f : G_{1} \to G_{2}$ be the bijection that Duplicator selects at the next round. Let $v \in V(\Gamma_{1})$ be a vertex on the same gadget as $u$, and suppose that for some single-support element $g_{v} \in G_{1}$ that is supported by $v$, that $f(g_{v})$ belongs to a different gadget than $f(g_{u}) = x$. Then Spoiler can win with $3$ pebbles and $O(1)$ rounds.

\item Suppose that $u \in V(\Gamma_{1})$ is an internal vertex on some gadget, and let $g_{u} \in G_{1}$ be a single-support element that is supported by $u$. Suppose that Duplicator selects a bijection $f : G_{1} \to G_{2}$ where $f(g_{u})$ is a single-support vertex is supported by an external vertex of some gadget. Then Spoiler can win with $3$ pebbles and $O(1)$ rounds.
\end{enumerate}
\end{lemma}

\begin{proof}
We proceed as follows.
\begin{enumerate}[label=(\alph*)]
\item Spoiler begins by pebbling $g_{u}g_{v} \mapsto f(g_{u}g_{v})$. Now if $|\supp(f(g_{u}g_{v}))| > 2$, Spoiler pebbles $g_{u} \mapsto x, g_{v} \mapsto y$ at the next two rounds. If $f(g_{u}g_{v}) \neq xy$, then Spoiler immediately wins. So suppose $f(g_{u}g_{v}) = xy$. As $|\supp(f(g_{u}g_{v}))| > 2$, either $|\supp(x)| > 1$ or $|\supp(y)| > 1$. Without loss of generality, suppose $|\supp(x)| > 1$. In this case, Spoiler wins by the argument in the proof of \Lem{lem:Centralizer}, reusing the pebbles on $g_{v}, g_{u}g_{v}$. 

So suppose now that $|\supp(f(g_{u}g_{v}))| = 2$. At the next two rounds, Spoiler pebbles $g_{u} \mapsto x, g_{v} \mapsto y$. Now by the CFI construction \cite{CFI}, the graphs $\Gamma_{1}, \Gamma_{2}$ have the property that for every $6$-cycle and every $8$-cycle, there exists a single gadget that contains said cycle. That is, no $6$-cycle and no $8$-cycle span two gadgets. Moreover, every pair of vertices lying on the same gadget lie on some common $6$-cycle or same $8$-cycle. Thus, Spoiler may reuse the pebble on $g_{u}g_{v}$ and trace along the cycle containing $g_{u}, g_{v}$ starting from $g_{u}$. Within at most $4$ additional rounds, Spoiler will have moved this third pebble to a neighbor of $g_{v}$, while the corresponding pebble will not be along a neighbor of $y$. Spoiler now wins.

\item Spoiler begins by pebbling $g_{v} \mapsto f(g_{v})$. Using a third pebble, we now proceed identically as in the second paragraph of part (a). The result follows.

\item Spoiler begins by pebbling $g_{u} \mapsto f(g_{u})$. We note that if $f(g_{u})$ is not supported by a single vertex, then by \Lem{lem:Centralizer}, Spoiler can win with $2$ additional pebbles and $2$ additional rounds. So suppose $f(g_{u})$ is supported by the vertex $x \in V(\Gamma_{2})$. Let $f' : G_{1} \to G_{2}$ be the bijection that Duplicator selects at the next round. Let $y\in V(\Gamma_{2})$ be an external vertex adjacent to $x$, and let $h_{y} \in G_{2}$ be a single-support group element that is supported by $y$. Let $g  \in G_{1}$ such that $f'(g) = h_{y}$. Spoiler pebbles $g \mapsto h_{y}$. 

We may again assume that $g$ has single support, or Spoiler wins in one additional pebble (beyond reusing the pebble on $g_{u}$) and two additionals rounds. By part (b), we may assume that $g_{u}, g$ belong to different gadgets, or Spoiler wins with $1$ additional pebble and $O(1)$ additional rounds. But as $g_{u}$ is supported by an internal vertex, so the vertices supporting $g_{u}, g$ are not adjacent in $\Gamma_{1}$. By Mekler's construction, this implies that $g_{u}, g$ do not commute. However, $f(g_{u}), h_{y}$ are adjacent and do commute. So Spoiler pebbles $g_{u}g$ and wins.
\end{enumerate}
\end{proof}

\noindent For convenience, we pull out the following construction.

\begin{definition} \label{def:GadgetProducts}
Let $G_{i}$ ($i = 1, 2$) be a CFI group. We first define a set $\mathcal{V}$ of vertices in $\Gamma_{1}$ as follows. For each gadget, we include a single, arbitrary internal vertex and all adjacent external vertices. Let $v \in G_{1}$ denote the ordered product of all the vertices in $\mathcal{V}$. 
\end{definition}

\noindent We now prove \Thm{thm:CFIWL}.

\begin{proof}[Proof of \Thm{thm:CFIWL}]
We follow the strategy of \cite[Theorem~6.1]{WLGroups}, carefully analyzing the number of rounds. We first define a set $\mathcal{V}$ of vertices in $\Gamma_{1}$ according to \Def{def:GadgetProducts}. Let $v \in G_{1}$ denote the ordered product of all the vertices in $\mathcal{V}$. By \Lem{lem:Support}, Duplicator must choose a bijection $f : G_{1} \to G_{2}$ in which $f(v)$ has the same support size as $v$. Otherwise, Spoiler can win with $3$ pebbles and $O(\log \log n)$ rounds. 

Spoiler begins by pebbling $v \mapsto f(v)$. Now by \Lem{lem:Support}, Duplicator must select bijections that map (setwise) $\supp(v) \mapsto \supp(f(v))$; otherwise, Spoiler can win with $2$ additional pebbles and $O(\log \log n)$ rounds. By \Lem{lem:GadgetStructure}, $\supp(f(v))$ must be composed exactly as $\supp(v)$; otherwise, Spoiler wins with $2$ additional pebbles and $O(1)$ rounds. That is, $\mathcal{V}' := \supp(f(v))$ must contain exactly one internal vertex and all adjacent external vertices from each gadget of $\Gamma_{2}$ (i.e., $\mathcal{V}'$ must also satisfy \Def{def:GadgetProducts}). 

Now in the proof of \cite[Theorem~6.1]{WLGroups}, Brachter \& Schweitzer showed that the induced subgraphs $\Gamma_{1}[\mathcal{V}]$ and $\Gamma_{2}[\mathcal{V}']$ have a different number of edges modulo $2$. In particular, $\Gamma_{1}[\mathcal{V}]$ and $\Gamma_{2}[\mathcal{V}']$ disagree in exactly one edge: the twisted link.

Let $f' : G_{1} \to G_{2}$ be the bijection that Duplicator selects at the next round. As the number of edges in $\Gamma_{1}[\mathcal{V}]$ and $\Gamma_{2}[\mathcal{V}']$ disagree, there exists a single-support vertex $g \in G_{1}$ such that the vertex supporting $g$ has degree in $\Gamma_{1}[\mathcal{V}]$ that is different than the degree of the vertex supporting $f'(g)$ in $\Gamma_{2}[\mathcal{V}']$. Spoiler pebbles $g \mapsto f'(g)$. At the next round, Duplicator must select a bijection $f'' : G_{1} \to G_{2}$ that maps some element $u$ of $\mathcal{V}$ that commutes with $g$ to some element $f''(u)$ of $\mathcal{V}'$ that does not commute with $f'(g)$ (or vice-versa). Spoiler pebbles $u \mapsto f''(u)$. Then at the next round, moves the pebble on $v$ to $gu$ and wins. In total, Spoiler used at most $3$ pebbles on the board and $O(\log \log n)$ rounds.

Furthermore, suppose that $\Gamma_{0}$ is identified by the graph $(3,r)$-WL. Brachter \& Schweitzer \cite{WLGroups} previously established that all single-support group elements of $G_{1}, G_{2}$ have centralizers of size $p^{4} \cdot |Z(G_{1})|$, and all other group elements have centralizers of size at most $p^{3} \cdot |Z(G_{1})|$. Now let $H$ be an arbitrary group, and suppose $3$-WL Version I fails to distinguish $G_{i}$ $(i = 1, 2)$ from $H$ in $\max\{r, O(\log \log n)\}$ rounds. Then $H$ has the same number of group elements with centralizers of size $p^{4} \cdot |Z(G_{1})|$. Furthermore, as $3$-WL Version I fails to distinguish $G_{i}$ $(i = 1, 2)$ from $H$ in $\max\{r, O(\log \log n)\}$ rounds, the induced commutation graph on these elements in $H/Z(H)$ is indistinguishable from $\Gamma_{i}$. Furthermore, by \Lem{lem:GadgetStructure}, $(3, O(1))$-WL Version I identifies internal vertices. So given $G_{i}$ ($i = 1, 2$), we can recover the base graph $\Gamma_{0}$. Furthermore, we can reconstruct the base graph $\Gamma$ underlying $H$. Precisely, any bijection $f : G_{i} \to H$ induces a bijection $\tilde{f} : V(\Gamma_{0}) \to V(\Gamma)$, and so we may simulate the $4$-pebble, $r$-round strategy to identify $\Gamma_{0}$ in the graph pebble game, by pebbling the appropriate elements of $G_{i}$ ($i = 1, 2$). But since $\Gamma_{0}$ is identified by the graph $(3,r)$-WL, we have that $\Gamma_{0} \cong \Gamma$. So $H$ is isomorphic to either $G_{1}$ or $G_{2}$. The result now follows.
\end{proof}

\subsection{Count-Free Strategy and the CFI Groups}

In this section, we establish the following theorem.

\begin{theorem} \label{thm:SecCountFreeCFI}
Let $\Gamma_{0}$ be a $3$-regular connected graph, and let $\Gamma_{1} := \text{CFI}(\Gamma_{0})$ and $\Gamma_{2} := \widetilde{\text{CFI}(\Gamma_{0})}$ by the corresponding CFI graphs. For $i = 1, 2$, denote $G_{i} := G_{\Gamma_{i}}$ to be the corresponding groups arising from Mekler's construction. 
\begin{enumerate}[label=(\alph*)]
\item The multiset of colors computed by the count-free $(O(1), O(\log \log n))$-WL Version I distinguishes $G_{1}$ from $G_{2}$. In particular, we can decide whether $G_{1} \cong G_{2}$ in $\beta_{1}\textsf{MAC}^{0}(\textsf{FOLL})$. 

\item If furthermore $\Gamma_{0}$ is identified by the graph count-free $(3, r)$-WL algorithm, then the multiset of colors computed by the count-free $(O(1), \max\{r, O(\log \log n)\})$-WL Version I will distinguish $G_{i}$ ($i = 1, 2$) from an arbitrary graph $H$. In particular, we can decide whether $G_{i} \cong H$ in $\beta_{1}\textsf{MAC}^{0}(\textsf{FOLL})$. 
\end{enumerate}
\end{theorem}

\noindent \\ We proceed similarly as in the case of counting WL. We begin with the following lemma.

\begin{lemma} \label{lem:CountFreeSingleSupport}
Let $G_{i}$ be a (twisted) CFI group ($i = 1, 2$). Let $u, v \in G_{i}$ where $|\supp(u)| = 1$ and $|\supp(v)| > 1$. Suppose that $u \mapsto v$ has been pebbled. Spoiler can win with $O(1)$ additional pebbles and $O(\log \log n)$ additional rounds.
\end{lemma}

\begin{proof}
Brachter \& Schweitzer \cite{WLGroups} previously established that $|C_{G_{i}}(u)| = p^{4} \cdot |Z(G_{i})|$ and $|C_{G_{i}}(v)| \leq p^{3} \cdot |Z(G)|$. Now by \cite[Lemma~4.7]{WLGroups} (recalled as \Lem{lem:BSFrattini}), we have that $Z(G_{i}) = \Phi(G_{i}) = [G_{i}, G_{i}]$. In particular, as $G_{i}$ is a class $2$ $p$-group of exponent $p > 2$, we have that $G_{i}/Z(G_{i})$ is elementary Abelian. So $C_{G_{i}}(u)/Z(G_{i}) \cong (\mathbb{Z}/p\mathbb{Z})^{4}$ and $C_{G_{i}}(v)/Z(G_{i}) \cong (\mathbb{Z}/p\mathbb{Z})^{d}$ for some $d \leq 3$. Spoiler now pebbles a representative $g_{1}, g_{2}, g_{3}, g_{4}$ of each coset for $C_{G_{i}}(u)/Z(G_{i})$. Let $h_{1}, h_{2}, h_{3}, h_{4}$ be the corresponding elements Duplicator pebbles. Necessarily, one such element either does not belong to $C_{G_{i}}(v)$ or belongs to $\langle h_{1}, h_{2}, h_{3} \rangle \cdot Z(G_{i})$. Without loss of generality, suppose this element is $h_{4}$. If $h_{4} \not \in C_{G_{i}}(v)$, Spoiler wins by pebbling $h_{4}, h_{4}v$ at the next two rounds. 

So suppose that $h_{4} \in C_{G_{i}}(v)$. Now each element of $\langle h_{1}, h_{2}, h_{3} \rangle \cdot Z(G_{i})$ can be written as $h_{1}^{e_{1}}h_{2}^{e_{2}}h_{3}^{e_{3}} \cdot z$ for some $z \in Z(G_{i})$. Using $3$ additional pebbles, Spoiler can pebble $g_{j}^{e_{j}}$ ($j = 1, 2, 3$). If Duplicator does not respond by pebbling $h_{j}^{e_{j}}$, then by \cite[Lemma~3.12]{GLWL1}, Spoiler can win with $O(1)$ additional pebbles in $O(\log \log n)$ rounds. Spoiler now pebbles $z \in Z(G_{i})$ such that $h_{4} = h_{1}^{e_{1}}h_{2}^{e_{2}}h_{3}^{e_{3}} \cdot z$ and wins. The result now follows.
\end{proof}

We now show that the count-free WL algorithm will distinguish group elements with different support sizes.

\begin{lemma} \label{lem:CountFreeSupport}
Let $G_{i}$ be a (twisted) CFI group ($i = 1, 2$). Let $u, v \in G_{i}$ where $|\supp(u)| \neq |\supp(v)| > 1$. Suppose that $u \mapsto v$ has been pebbled. Spoiler can win with $4$ additional pebbles and $O(\log \log n)$ additional rounds.
\end{lemma}

\begin{proof}
Aside from Spoiler selecting an element $x$ to pebble, the proof of \Lem{lem:Support} did not rely on Duplicator selecting a bijection at each round. Thus, as $u \mapsto v$ has been pebbled, we may proceed identically as in \Lem{lem:Support}. The result now follows.
\end{proof}

Our next goal is to show that count-free WL can detect the gadget structure of the underlying CFI graphs.

\begin{lemma} \label{lem:CountFreeGadget}
Let $u \in V(\Gamma_{i})$ ($i = 1, 2$), and let $g_{u} \in G_{i}$ be a single-support element that is supported by $u$.
\begin{enumerate}[label=(\alph*)]
\item Let $v \in V(G_{i})$ be on the same gadget as $u$, and let $g_{v} \in G_{i}$ be a single-support element that is supported by $v$. Suppose that $(g_{u}, g_{v}) \mapsto (h_{u}, h_{v})$ have been pebbled, and that $h_{u}h_{v}$ is not supported by exactly two vertices on the same gadget. Then Spoiler can win with $O(1)$ additional pebbles and $O(1)$ rounds.

\item Suppose now that $u$ is an internal vertex. Suppose that $g_{u} \mapsto h_{u}$ has been pebbled, and that $h_{u}$ is a single-support element supported by some $x$ that is an external vertex. Then Spoiler can win with $O(1)$ additional pebbles and $O(1)$ rounds.
\end{enumerate}
\end{lemma}

\begin{proof}
We proceed as follows.
\begin{enumerate}[label=(\alph*)]
\item We note that if $|\supp(h_{u}h_{v})| \neq 2$, then either $h_{u}$ or $h_{v}$ are not single-support elements. In this case, Spoiler can win with $O(1)$ additional pebbles and $O(1)$ additional rounds by \Lem{lem:CountFreeSingleSupport}. So suppose $|\supp(h_{u}h_{v})| = 2$. Let $\supp(h_{u}) = \{x\}$ and $\supp(h_{v}) = \{y\}$. Suppose that $x, y$ belong to different gadgets. Now the CFI graphs have the property that two vertices belong to a common gadget if and only if they are on common $6$-cycle or common $8$-cycle \cite{CFI}. Using two additional pebbles and $O(1)$ additional rounds, Spoiler wins by tracing around the cycle containing $u, v$. 

\item As $x$ is an external vertex, $x$ is adjacent to some other external vertex $y$. Let $h_{y}$ be a single-support element that is supported by $y$. Spoiler pebbles $h_{y}$, and Duplicator responds by pebbling some single-support element $g_{v}$ that is supported by the vertex $v$. If $uv \not \in E(\Gamma_{i})$, Spoiler immediately wins, as $h_{u}h_{y}$ commute, and $g_{u}g_{v}$ do not. So suppose $uv \in E(\Gamma_{i})$. But as $u$ is internal, $u, v$ belong to the same gadget, while $x, y$ do not. Spoiler now wins by part (a). 
\end{enumerate}
\end{proof}

\noindent We now establish the relationship between the group elements arising from the construction in \Def{def:GadgetProducts} and the induced subgraphs from $\mathcal{V}$. 

\begin{lemma} \label{lem:GadgetProducts}
Let $v \in G_{i}$ ($i = 1, 2$) such that $v$ satisfies the construction in \Def{def:GadgetProducts}. We have the following.
\begin{enumerate}[label=(\alph*)]
\item Let $v' \in G_{i'}$ ($i' = 1, 2$) such that $v'$ is not constructed according to \Def{def:GadgetProducts}. Then the count-free $(O(1), O(\log \log n))$-WL Version I will distinguish $v$ from $v'$.

\item Let $v' \in G_{i'}$ ($i' = 1, 2$) such that $v'$ is constructed according to \Def{def:GadgetProducts}. Let $u \in \supp(v)$, and let $g_{u} \in G_{i}$ be a single-support element that is supported by $u$. Let $h \in G_{i'}$. If $h$ is not a single-support element satisfying $\supp(h) \subseteq \supp(v')$, then the count-free $(O(1), O(\log \log n))$-WL Version I will distinguish $(v, g_{u})$ from $(v', h)$.

\item Let $v', g_{u}, h$ be as defined in part (b). Let $\supp(h) = \{ u'\}$, and relabel $h_{u'} := h$. Suppose that $|N(u) \cap \supp(v)| \neq |N(u') \cap \supp(v')|$. That is, suppose that the degree of $u$ in $\Gamma_{i}[\supp(v)]$ is different than the degree of $u'$ in $\Gamma_{i'}[\supp(v')]$. Then the count-free $(O(1), O(\log \log n))$-WL Version I will distinguish $(v, g_{u})$ from $(v', h_{u'})$.
\end{enumerate}
\end{lemma}

\begin{proof}
We proceed as follows.
\begin{enumerate}[label=(\alph*)]
\item By \Lem{lem:CountFreeSupport}, if $v' \in G_{i}$ ($i = 1, 2$) satisfies $|\supp(v')| \neq |\supp(v)|$, then the $(O(1), O(\log \log n))$-WL Version I algorithm will distinguish $v$ and $v'$.

Now suppose $\supp(v')$ contains two vertices $a, b$ on the same gadget. We claim that the count-free $(O(1), O(\log \log n))$-WL Version I will distinguish $v$ from $v'$. Consider the pebble game, starting from the configuration $v \mapsto v'$. Spoiler pebbles group elements $h_{a}, h_{b}$ supported by $a, b$ respectively. Duplicator responds by pebbling $g_{x}, g_{y}$. By \Lem{lem:CountFreeSupport}, we may assume that $g_{x}, g_{y}$ are single-support elements; otherwise, Spoiler wins with $O(1)$ pebbles and $O(1)$ rounds. Let $x, y$ be the vertices of $\Gamma_{i}$ supporting $g_{x}, g_{y}$ respectively. We may assume that $x, y \in \supp(v)$. Otherwise, by \Lem{lem:CountFreeSupport}, Spoiler wins with $O(1)$ pebbles and $O(\log \log n)$ rounds. By construction of $v$, $x, y$ lie on different gadgets. Thus, by \Lem{lem:CountFreeGadget}, Spoiler can win with $O(1)$ additional pebbles and $O(1)$ additional rounds. It now follows that any group element $v'$ whose support does not consist of a single arbitrary internal vertex from each gadget and all  external vertices adjacent to the internal vertices selected, that the count-free $(O(1), O(\log \log n))$-WL Version I will distinguish $v$ and $v'$.

\item If $h$ is not single support or $\supp(h) \not \subseteq \supp(v')$, then by \Lem{lem:CountFreeSupport}, the count-free $(O(1), O(\log \log n))$-WL Version I algorithm will distinguish $(v, g_{u})$ from $(v', h)$.

\item By \Lem{lem:CountFreeSupport}, if $\supp(h_{u}) \not \subseteq \supp(v')$, then the count-free $(O(1), O(\log \log n))$-WL Version I will distinguish $(v, g_{u})$ and $(v', h_{u'})$. Now by the CFI construction \cite{CFI}, both $\Gamma_{i}, \Gamma_{i'}$ have maximum degree at most $4$. So with $O(1)$ pebbles, Spoiler can pebble the neighbors of $u$ in $\Gamma_{i}$. By construction, the adjacency relation in $\Gamma_{i}, \Gamma_{i'}$ determines the commutation relation in $G_{i}, G_{i'}$. Thus, only $O(1)$ additional pebbles and $O(1)$ additional rounds are needed to determine commutation. Thus, the count-free $(O(1), O(\log \log n))$-WL Version I will distinguish $(v, g_{u})$ and $(v', h_{u'})$.
\end{enumerate}
\end{proof}

\noindent We now prove \Thm{thm:SecCountFreeCFI}.

\begin{proof}[Proof of \Thm{thm:SecCountFreeCFI}]
We proceed similarly as in the proof of \Thm{thm:CFIWL}. 

\begin{enumerate}[label=(\alph*)]
\item Let $v \in G_{1}$ be defined as in \Def{def:GadgetProducts}. Now suppose that $G_{1} \not \cong G_{2}$. Now take an arbitrary $v' \in G_{2}$ such that $v'$ is defined according to \Def{def:GadgetProducts}. Again take $\mathcal{V} = \supp(v)$ and $\mathcal{V}' = \supp(v')$. Brachter \& Schweitzer \cite[Proof of Theorem~6.1]{WLGroups} argued that the induced subgraphs $\Gamma_{1}[\mathcal{V}]$ and $\Gamma_{2}[\mathcal{V}']$ have different degree sequences. 

Now suppose that $u \in \mathcal{V}$ and $u' \in \mathcal{V}'$ have different degrees. Let $g_{u} \in G_{1}$ be a single-support element supported by $u$, and let $h_{u'} \in G_{2}$ be a single-support element supported by $u'$. As $u, u'$ have different degrees, we have by \Lem{lem:GadgetProducts}(c) that the count-free $(O(1), O(\log \log n))$-WL Version I will distinguish $(v, g_{u})$ from $(v', h_{u'})$. In particular, it follows that the multiset of colors produced by the count-free $(O(1), O(\log \log n))$-WL Version I algorithm will be different for $G_{1}$ than for $G_{2}$.

By Grohe \& Verbitsky, we have that the count-free $(O(1), O(\log \log n))$-WL Version I can be implemented using an $\textsf{FOLL}$ circuit. We will now use a $\beta_{1}\textsf{MAC}^{0}$ circuit to distinguish $G_{1}$ from $G_{2}$. Using $O(\log n)$ non-deterministic bits, we guess a color class $C$ where the multiplicity differs. At each iteration, the parallel WL implementation due to Grohe \& Verbitsky records indicators as to whether two $k$-tuples receive the same color. As we have already run the count-free WL algorithm, we may in $\textsf{AC}^{0}$ decide whether two $k$-tuples have the same color. For each $k$-tuple in $G_{1}^{k}$ having color $C$, we feed a $1$ to the $\textsf{Majority}$ gate. For each $k$-tuple in $G_{2}^{k}$ having color $C$, we feed a $0$ to the $\textsf{Majority}$ gate. The result now follows.

\item Suppose furthermore that the base graph $\Gamma_{0}$ is identified by the count-free $(O(1), r)$-WL algorithm for graphs. Brachter \& Schweitzer \cite{WLGroups} previously established that single-support elements of $G_{1}, G_{2}$ have centralizers of size $p^{4} \cdot |Z(G_{1})|$, and all other group elements have centralizers of size at most $p^{3} \cdot |Z(G_{1})|$. By \Lem{lem:CountFreeSingleSupport}, the count-free $(O(1), O(1))$-WL Version I will distinguish in $G_{i}$ ($i = 1, 2$) single-support group elements from those group elements $g$ with $|\supp(g)| > 1$. Now let $H$ be an arbitrary group, and suppose the multiset of colors produced by the count-free $(O(1), O(\log \log n))$-WL Version I is the same for $G_{i}$ ($i = 1, 2$) as for $H$. Then $G_{i}$ ($i = 1, 2$) and $H$ has the same number of elements of order $p^{4} \cdot |Z(G_{1})|$. Furthermore, as the multiset of colors arising from the count-free $(O(1), O(\log \log n))$-WL Version I fails to distinguish $G_{i}$ ($i = 1, 2$) and $H$, the induced commutation graph on these elements in $H/Z(H)$ is indistinguishable from $\Gamma_{i}$. Furthermore, by \Lem{lem:CountFreeGadget}(b), the count-free $(O(1), O(1))$-WL Version I will distinguish internal and external vertices. So given $G_{i}$ ($i = 1, 2$), we can reconstruct the base graph $\Gamma_{0}$. Furthermore, we can reconstruct the base graph $\Gamma$ underlying $H$. Precisely, the vertices of $\Gamma_{0}$ correspond to gadgets of $\Gamma_{i}$. Now the count-free WL Version I for groups can simulate the count-free WL for graphs in the following manner. When Spoiler or Duplicator pebble a single-support element of $G_{i}$, that induces placing a pebble on the corresponding vertex $v$ of $\Gamma_{i}$. In turn, this induces placing a pebble on the vertex corresponding to the gadget of $\Gamma_{0}$ containing $v$. So we may simulate the $(3,r)$-pebble strategy to identify $\Gamma_{0}$ in the graph pebble game, by pebbling the appropriate elements of $G_{i}$ ($i = 1, 2$). But since $\Gamma_{0}$ is identified by the graph $(3,r)$-WL, we have that $\Gamma_{0} \cong \Gamma$. So $H$ is isomorphic to either $G_{1}$ or $G_{2}$. The result now follows.
\end{enumerate}
\end{proof}

\section{Coprime Extensions}

In this section, we consider groups that admit a Schur--Zassenhaus decomposition of the form $G = H \ltimes N$, where $N$ is Abelian, $H$ is $O(1)$-generated and solvable with solvability class $\poly \log \log n$, and $|H|$ and $|N|$ are coprime. Precisely, we establish the following.

\begin{theorem} \label{thm:QST}
Let $G = H \ltimes N$, where $H$ is $O(1)$-generated, $N$ is Abelian, and $\text{gcd}(|H|, |N|) = 1$. Let $K$ be arbitrary. 
\begin{enumerate}[label=(\alph*)]
\item Fix $c \geq 0$. If $H$ is solvable with solvability class $O((\log \log n)^{c})$, then can decide whether $G$ and $K$ are isomorphic in $\beta_{1}\textsf{MAC}^{0}(\textsf{FO}((\log \log n)^{c+1}))$. 

\item If $H$ is a finite simple group, we can decide isomorphism in $\beta_{1}\textsf{MAC}^{0}(\textsf{FOLL})$.
\end{enumerate}
\end{theorem}

\noindent 
\begin{remark}
Prior to this paper, isomorphism testing of the groups considered in \Thm{thm:QST} were first shown to be in $\textsf{P}$ by Qiao, Sarma, \& Tang \cite{QST11}. Grochow \& Levet \cite{GLWL1} subsequently improved this bound to $\textsf{L}$ using Weisfeiler--Leman. We note that Qiao, Sarma, \& Tang and Grochow \& Levet considered a more general family, where where $H$ was only assumed to be $O(1)$-generated.
\end{remark}


\subsection{Additional preliminaries for groups with Abelian normal Hall subgroup}
Here we recall additional preliminaries needed for our algorithm in the next section. None of the results in this section are new, though in some cases we have rephrased the known results in a form more useful for our analysis. We recall the preliminaries here and refer the reader to \cite[Section~3.1]{GLWL1}  for the proofs.

A \emph{Hall} subgroup of a group $G$ is a subgroup $N$ such that $|N|$ is coprime to $|G/N|$. When a Hall subgroup is normal, we refer to the group as a coprime extension. Coprime extensions are determined entirely by their actions: 

\begin{lemma}[{Taunt \cite{Taunt1955}}]  \label{LemmaSemidirect}
Let $G = H \ltimes_{\theta} N$ and $\hat{G} = \hat{H} \ltimes_{\hat{\theta}} \hat{N}$. If $\alpha : H \to \hat{H}$ and $\beta : N \to \hat{N}$ are isomorphisms such that for all $h \in H$ and all $n \in N$,
\[
\hat{\theta}_{\alpha(h)}(n) = (\beta \circ \theta_{h} \circ \beta^{-1})(n),
\]
then the map $(h, n) \mapsto (\alpha(h), \beta(n))$ is an isomorphism of $G \cong \hat{G}$. Conversely, if $G$ and $\hat{G}$ are isomorphic and $|H|$ and $|N|$ are coprime, then there exists an isomorphism of this form.
\end{lemma}

\begin{remark}
\Lem{LemmaSemidirect} can be significantly generalized to arbitrary extensions where the subgroup is characteristic. The setting where the characteristic subgroup is Abelian has been quite useful for practical algorithms in isomorphism testing (see, e.g., \cite{Holt2005HandbookOC}) and is standard in group theory. More generally, the equivalence of group extensions deals with both \algprobm{Action Compatibility} and \algprobm{Cohomology Class Isomorphism}. For generalizations of cohomology to the setting of non-Abelian group coefficients, see for instance \cite{dedecker, inassaridze}. Grochow \& Qiao re-derived parts of this work in the setting of $H^{2}$---the cohomology most immediately relevant to group extensions and the isomorphism problem--- generalizing Taunt's lemma. They also provided algorithmic applications \cite[Lemma~2.3]{GQcoho}. In the setting of coprime extensions, the Schur--Zassenhaus Theorem provides that the cohomology is trivial. Thus, in our setting we need only consider \algprobm{Action Compatibility}.
\end{remark}

A $\Z H$-module is an abelian group $N$ together with an action of $H$ on $N$, given by a group homomorphism $\theta\colon H \to \Aut(N)$. Sometimes we colloquially refer to these as ``$H$-modules.'' A \emph{submodule} of an $H$-module $N$ is a subgroup $N' \leq N$ such that the action of $H$ on $N'$ sends $N'$ into itself, and thus the restriction of the action of $H$ to $N'$ gives $N'$ the structure of an $H$-module compatible with that on $N$. Given a subset $S \subseteq N$, the smallest $H$-submodule containing $S$ is denoted $\langle S \rangle_H$, and is referred to as the $H$-submodule \emph{generated by} $S$. An $H$-module generated by a single element is called \emph{cyclic}. Note that a cyclic $H$-module $N$ need not be a cyclic Abelian group.

Two $H$-modules $N,N'$ are isomorphic (as $H$-modules), denoted $N \cong_H N'$, if there is a group isomorphism $\varphi\colon N \to N'$ that is $H$-equivariant, in the sense that $\varphi(\theta(h)(n)) = \theta'(h)(\varphi(n))$ for all $h \in H, n \in N$. An $H$-module $N$ is \emph{decomposable} if $N \cong_H N_1 \oplus N_2$ where $N_1, N_2$ are nonzero $H$-modules (and the direct sum can be thought of as a direct sum of Abelian groups); otherwise $N$ is \emph{indecomposable}. An equivalent characterization of $N$ being decomposable is that there are nonzero $H$-submodules $N_1, N_2$ such that $N = N_1 \oplus N_2$ as Abelian groups (that is, $N$ is generated as a group by $N_1$ and $N_2$, and $N_1 \cap N_2 = 0$). The Remak--Krull--Schmidt Theorem says that every $H$-module decomposes as a direct sum of indecomposable modules, and that the multiset of $H$-module isomorphism types of the indecomposable modules appearing is independent of the choice of decomposition, that is, it depends only on the $H$-module isomorphism type of $N$. We may thus write 
\[
N \cong_H N_1^{\oplus m_1} \oplus N_2^{\oplus m_2} \oplus \dotsb \oplus N_k^{\oplus m_k}
\]
unambiguously, where the $N_i$ are pairwise non-isomorphic indecomposable $H$-modules. When we refer to the multiplicity of an indecomposable $H$-module as a direct summand in $N$, we mean the corresponding $m_i$.\footnote{For readers familiar with (semisimple) representations over fields, we note that the multiplicity is often equivalently defined as $\dim_\F \text{Hom}_{\F H}(N_i, N)$. However, when we allow $N$ to be an Abelian group that is not elementary Abelian, we are working with $(\Z/p^k \Z)[H]$-modules, and the characterization in terms of hom sets is more complicated, because one indecomposable module can be a submodule of another, which does not happen with semisimple representations. }

The version of Taunt's Lemma that will be most directly useful for us is:
\begin{lemma} \label{LemTaunt}
Let $G_i = H \ltimes_{\theta_i} N$ for $i=1,2$ be two semi-direct products with $|H|$ coprime to $|N|$. Then $G_1 \cong G_2$ if and only if there is an automorphism $\alpha \in \Aut(H)$ such that each indecomposable $\Z H$-module appears as a direct summand in $(N, \theta_1)$ and in $(N, \theta_2 \circ \alpha)$ with the same multiplicity.
\end{lemma}


\begin{proof}
See \cite[Lemma~3.2]{GLWL1}.
\end{proof}

When $N$ is elementary Abelian the following lemma is not necessary. A $(\Z/p^k \Z)[H]$-module is an $\Z H$-module $N$ where the exponent of $N$ (the LCM of the orders of the elements of $N$) divides $p^k$. 

\begin{lemma}[{see, e.\,g., Thev\'{e}naz \cite{Thevenaz}}] \label{LemThevenaz}
Let $H$ be a finite group. If $p$ is coprime to $|H|$, then any indecomposable $(\Z/p^k \Z)[H]$-module is generated (as an $H$-module) by a single element.
\end{lemma}

\begin{proof}
See \cite[Lemma~3.3]{GLWL1}.
\end{proof}

\subsection{Coprime Extensions with an $O(1)$-Generated Complement}

Our goal is to show that if $G = H \ltimes N$ and $K \not \cong G$, then the multiset of colors computed by the count-free $(O(1), \poly \log \log n)$-WL Version I algorithm will differ for $G$ and $K$. We first show that if $G = H \ltimes N$ and $K$ does not admit a decomposition of the form $H \ltimes N$, then the multiset of colors computed by the count-free $(O(1), O(\log \log n))$-WL Version I will differ for $G$ and $K$. 

So now suppose that $K$ also admits a decomposition of the form $H \ltimes N$. For groups that decompose as a coprime extension of $H$ and $N$, the isomorphism type is completely determined by the multiplicities of the indecomposable $H$-module direct summands (\Lem{LemTaunt}). Thus, if $G \not \cong K$, then for any fixed generators $h_{1}, \ldots, h_{d}$ of $H \leq G$ and $h_{1}', \ldots, h_{d}'$ of $H \leq K$, there exists an indecomposable $\langle h_{1}, \ldots, h_{d} \rangle$-module in $N \leq G$ whose isomorphism type appears with multiplicity different than in $N \leq K$. We leverage the fact that, in the coprime case, indecomposable $H$-modules are generated by single elements (\Lem{LemThevenaz}), which allows us to utilize the constant-dimensional count-free WL algorithm.

\begin{lemma} \label{CheckCoprime}
Let $G = H \ltimes N$, where $H$ is $O(1)$-generated, $N$ is Abelian, and $\text{gcd}(|H|, |N|) = 1$. Let $K$ be arbitrary. 
\begin{enumerate}[label=(\alph*)]
\item Suppose that $H$ is solvable with solvability class $O((\log \log n)^{c})$. If $K$ does not decompose as $H \ltimes N$, then the multiset of colors computed by the count-free $(O(1), O(\log \log n)^{c+1})$-WL Version I will be different for $G$ than for $K$. 

\item Suppose that $H$ is a finite simple group. If $K$ does not decompose as $H \ltimes N$, then the multiset of colors computed by the count-free $(O(1), O(\log \log n))$-WL Version I will be different for $G$ than for $K$.
\end{enumerate} 

\noindent 
\end{lemma}

\begin{proof}
Suppose $H \not \leq K$. If $H$ is solvable with solvability class $O((\log \log n)^{c})$, then by \Thm{thm:BoundedGenSolvable}, Spoiler can pebble generators for $H$ and win with $O(1)$ additional pebbles and $O((\log \log n)^{c+1})$ additional rounds. In the case when $H$ is simple, the result still holds via \Cor{cor:FOLLSimple} (using \Thm{thm:BabaiKantorLubotzky}).

Now by \cite[Lemma~7.12]{GLWL1}, we note that if the multiset of orders in $G$ differs from $K$, then the multiset of colors computed by the count-free $(O(1), O(\log \log n))$-WL will distinguish $G$ from $K$. Now as $N \leq G$ is the set of elements of order dividing $|N|$, we have that if multiset of colors computed by the count-free $(O(1), O(\log \log n))$-WL fails to distinguish $G$ and $K$, then $K$ set $N_{K}$ of order $|N|$ containing all elements of $K$ with order dividing $|N|$. Furthermore, the multiset of orders in $N_{K}$ must be identical to $N$. Now if $N_{K}$ is not a subgroup of $K$, then there exist $a, b \in N_{K}$ such that $ab \not \in N_{K}$. So again as count-free WL detects orders in $O(\log \log n)$ rounds (see \cite[Lemma~7.12]{GLWL1}), count-free WL will distinguish $(a,b,ab)$ from any pair $a', b', a'b' \in N$ after $O(\log \log n)$ rounds. In this case, the multiset of colors for $G$ and $K$ will differ after $O(\log \log n)$ rounds. Thus, we may assume that $N_{K}$ is a subgroup of $K$. As $N_{K}$ is the set of all elements of order dividing $|N|$, we have that $N_{K}$ is characteristic and thus normal.

It remains to show that if $N_{K}$ is not Abelian, then the multiset of colors computed by the count-free $(O(1), O(\log \log n))$-WL fails to distinguish $G$ and $K$. We consider the count-free pebble game. Spoiler begins in the first two rounds by pebbling two elements $h_{1}, h_{2} \in N_{K}$ that don't commute. Duplicator responds by pebbling $g_{1}, g_{2}$. We note that if $|g_{i}| \neq |h_{i}|$, then by \cite[Lemma~7.12]{GLWL1} Spoiler can win with $O(1)$ additional pebbles and $O(\log \log n)$ rounds. So assume that $|g_{i}| = |h_{i}|$. Thus, $g_{1}, g_{2} \in N$. As $N$ is Abelian, $[g_{1}, g_{2}] = 1$. So Spoiler pebbles $g_{1}g_{2}$ and wins. The result now follows.
\end{proof}

\begin{theorem} \label{thm:QSTSec5}
Let $G = H \ltimes N$, where $H$ is $O(1)$-generated, $N$ is Abelian, and $\text{gcd}(|H|, |N|) = 1$. Let $K$ be arbitrary. 
\begin{enumerate}[label=(\alph*)]
\item Suppose that $H$ is solvable with solvability class $O((\log \log n)^{c})$.  If $G \not \cong K$, then the multiset of colors computed by the count-free $(O(1), O((\log \log n)^{c+1}))$-WL Version I will be different for $G$ than for $K$.

\item Suppose that $H$ is a finite simple group. If $G \not \cong K$, then the multiset of colors computed by the count-free $(O(1), O(\log \log n))$-WL Version I will be different for $G$ than for $K$.
\end{enumerate} 
\noindent
\end{theorem}

\begin{proof}
By \Lem{CheckCoprime}, we may assume that $K = H \ltimes N$. Let $h_{1}, \ldots, h_{d} \in G$ be generators for a prescribed copy of $H$, and let $h_{1}', \ldots, h_{d}' \in K$ be corresponding generators for a copy of $H$ in $K$. Now let $n$ be an element of $N \leq G$, and let $n'$ be an element of $N \leq K$. Suppose $\langle n \rangle_{\langle h_{1}, \ldots, h_{d} \rangle}$ and $\langle n' \rangle_{\langle h_{1}' ,\ldots, h_{d}'\rangle}$ are inequivalent $H$-modules. Now $\langle n \rangle_{\langle h_{1}, \ldots, h_{d} \rangle}$ is generated as a group by the $\langle h_{1}, \ldots, h_{d} \rangle$-conjugates of $N$. Once $h_{1}, \ldots, h_{d}$ have been individualized, the elements of $\langle h_{1}, \ldots, h_{d} \rangle$ receive unique colors after (a) $O((\log \log n)^{c+1})$ additional rounds (see \Thm{thm:BoundedGenSolvable} for when $H$ is solvable with solvability class $O((\log \log n)^{c})$, and (b) $O(\log \log n)$ rounds when $H$ is simple \Cor{cor:FOLLSimple} using \Thm{thm:BabaiKantorLubotzky}). In particular, this implies that count-free WL will detect $\langle h_{1}, \ldots, h_{d} \rangle$ after (a) $O((\log \log n)^{c+1})$ additional rounds (see \Thm{thm:BoundedGenSolvable} for when $H$ is solvable with solvability class $O((\log \log n)^{c}))$, and (b) $O(\log \log n)$ rounds when $H$ is simple (see \Cor{cor:FOLLSimple} using \Thm{thm:BabaiKantorLubotzky}).

Let $n_{1}, \ldots, n_{\ell}$ be the $\langle h_{1}, \ldots, h_{d} \rangle$-conjugates of $n$. By \cite[Proposition~7.13]{GLWL1}, we have that for all $i, j$, $n_{i}^{j}$ will receive a unique color after $O(\log \log n)$ rounds. As $N$ is Abelian, each element of $\langle n \rangle_{\langle h_{1}, \ldots, h_{d} \rangle}$ can be written as an element of the form $n_{1}^{e_{1}} \cdots n_{\ell}^{e_{\ell}}$, where $e_{i} \in \{1, \ldots, |n_{\ell}|-1\}$. As $\ell \leq \lceil \log |G| \rceil$, we have by \cite[Lemma~7.3(c)]{GLWL1} that each element of $\langle n \rangle_{\langle h_{1}, \ldots, h_{d} \rangle}$ will receive a unique color after an additional $O(\log \log n)$ rounds. Thus, by considering the count-free pebble game from the initial configuration $(h_{1}, \ldots, h_{d}) \mapsto (h_{1}', \ldots, h_{d}')$, we may assume- at the cost of $O(1)$ pebbles and $O(\log \log n)$ rounds- that Spoiler and Duplicator respect the map $\langle n \rangle_{\langle h_{1}, \ldots, h_{d} \rangle} \mapsto \langle n' \rangle_{\langle h_{1}, \ldots, h_{d}\rangle}$ induced by $(n, h_{1}, \ldots, h_{d}) \mapsto (n', h_{1}', \ldots, h_{d}')$. 

Now as $G, K$ are non-isomorphic groups of the form $H \ltimes N$, they differ only in their actions. Now the actions are determined by the multiset of indecomposable $H$-modules in $N$. As $|H|, |N|$ are coprime, we have by \Lem{LemThevenaz} that the indecomposable $H$-modules are cyclic. As $G \not \cong K$, we have by \Lem{LemTaunt} that for any fixed choice of generators $h_{1}, \ldots, h_{d}$ for $H \leq G$ and any fixed choice of generators $h_{1}', \ldots, h_{d}'$ for $H \leq K$, there exists an $n \in N \leq G$ such that the multiplicity of the isomorphism class of $\langle n \rangle_{\langle h_{1}, \ldots, h_{d} \rangle}$ is different in $G$ than in $K$. As $h_{1}, \ldots, h_{d}$ and $h_{1}', \ldots, h_{d}'$ were arbitrary, it follows that the multiset of colors computed by the count-free $O(1)$-WL will be different for $G$ than for $K$ after $O((\log \log n)^{c+1})$ rounds for (a) and $O(\log \log n)$ rounds for (b).
\end{proof}

\begin{corollary}
Let $G = H \ltimes N$, where $H$ is $O(1)$-generated, $N$ is Abelian, and $\text{gcd}(|H|, |N|) = 1$. Let $K$ be arbitrary. 
\begin{enumerate}[label=(\alph*)]
\item Suppose that $H$ is solvable with solvability class $O((\log \log n)^{c})$.  We can decide whether $G \cong K$ in $\beta_{1}\textsf{MAC}^{0}(\textsf{FO}((\log \log n)^{c+1}))$.

\item Suppose that $H$ is a finite simple group. We can decide whether $G \cong K$ in $\beta_{1}\textsf{MAC}^{0}(\textsf{FOLL})$.
\end{enumerate} 
\noindent 
\end{corollary}

\begin{proof}
We apply the count-free $O(1)$-WL Version I algorithm for $O((\log \log n)^{c+1})$ rounds in (a) and $O(\log \log n)$ rounds for (b). By the parallel WL implementation due to Grohe \& Verbitsky \cite{GroheVerbitsky}, each round of count-free WL Version I is $\textsf{AC}^{0}$-computable. Now by \Thm{thm:QSTSec5}, if $G \not \cong K$, then in (a) the multiset of colors computed by $(O(1), O((\log \log n)^{c+1}))$-WL Version I will differ between $G$ and $K$. For (b), the multiset of colors computed by $(O(1), O(\log \log n))$-WL Version I will differ between $G$ and $K$. We will now use a $\beta_{1}\textsf{MAC}^{0}$ circuit to distinguish $G$ from $K$. Using $O(\log n)$ non-deterministic bits, we guess a color class $C$ where the multiplicity differs. At each iteration, the parallel WL implementation due to Grohe \& Verbitsky records indicators as to whether two $k$-tuples receive the same color. As we have already run the count-free WL algorithm, we may in $\textsf{AC}^{0}$ decide whether two $k$-tuples have the same color. For each $k$-tuple in $G^{k}$ having color $C$, we feed a $1$ to the $\textsf{Majority}$ gate. For each $k$-tuple in $K^{k}$ having color $C$, we feed a $0$ to the $\textsf{Majority}$ gate. The result now follows.
\end{proof}

\section{Direct Products of Non-Abelian Simple Groups}

In this section, we establish the following.

\begin{theorem} \label{thm:DirectProductSimple}
Let $G$ be a direct product of non-Abelian finite simple groups, and let $H$ be arbitrary. We can decide isomorphism between $G$ and $H$ in $\beta_{1}\textsf{MAC}^{0}(\textsf{FOLL})$.
\end{theorem}

\begin{remark}
Isomorphism testing of direct products of non-Abelian finite simple groups was known to be in $\textsf{P}$ via polynomial-time direct product decomposition results \cite{WilsonDirectProductsArxiv, KayalNezhmetdinov} and CFSG (in particular, the result that finite simple groups are $2$-generated). Brachter \& Schweitzer \cite[Lemmas 5.20 \& 5.21]{BrachterSchweitzerWLLibrary} established the analogous result for Weisfeiler--Leman Version II. A careful analysis of their work shows that only $O(1)$ rounds are required, which places isomorphism testing into $\textsf{L}$.
\end{remark}

\subsection{Preliminaries}
We note that direct products of finite non-Abelian simple groups have no Abelian normal subgroups (i.e., are semisimple). Thus, we recall some facts about semisimple groups from \cite{BCGQ, GLWL1}. As a semisimple group $G$ has no Abelian normal subgroups, we have that $\Soc(G)$ is the direct product of non-Abelian simple groups. The conjugation action of $G$ on $\Soc(G)$ permutes the direct factors of $\Soc(G)$. So there exists a faithful permutation representation $\alpha : G \to G^{*} \leq \Aut(\Soc(G))$. $G$ is determined by $\Soc(G)$ and the action $\alpha$. Let $H$ be a semisimple group with the associated action $\beta : H \to \text{Aut}(\Soc(H))$. We have that $G \cong H$ precisely if $\Soc(G) \cong \Soc(H)$ via an isomorphism that makes $\alpha$ equivalent to $\beta$. 

We now introduce the notion of permutational isomorphism, which is our notion of equivalence for $\alpha$ and $\beta$. Let $A$ and $B$ be finite sets, and let $\pi : A \to B$ be a bijection. For $\sigma \in \text{Sym}(A)$, let $\sigma^{\pi} \in \text{Sym}(B)$ be defined by $\sigma^{\pi} := \pi^{-1}\sigma \pi$. For a set $\Sigma \subseteq \text{Sym}(A)$, denote $\Sigma^{\pi} := \{ \sigma^{\pi} : \sigma \in \Sigma\}$. Let $K \leq \text{Sym}(A)$ and $L \leq \text{Sym}(B)$ be permutation groups. A bijection $\pi : A \to B$ is a \textit{permutational isomorphism} $K \to L$ if $K^{\pi} = L$.

The following lemma, applied with $R = \Soc(G)$ and $S = \Soc(H)$, precisely characterizes semisimple groups \cite{BCGQ}.      
 
\begin{lemma}[{\cite[Lemma 3.1]{BCGQ}}] \label{CharacterizeSemisimple}
Let $G$ and $H$ be groups, with $R \triangleleft G$ and $S \triangleleft H$ groups with trivial centralizers. Let $\alpha : G \to G^{*} \leq \Aut(R)$ and $\beta : H \to H^{*} \leq \Aut(S)$ be faithful permutation representations of $G$ and $H$ via the conjugation action on $R$ and $S$, respectively. Let $f : R \to S$ be an isomorphism. Then $f$ extends to an isomorphism $\hat{f} : G \to H$ if and only if $f$ is a permutational isomorphism between $G^{*}$ and $H^{*}$; and if so, $\hat{f} = \alpha f^{*} \beta^{-1}$, where $f^{*} :  G^{*} \to H^{*}$ is the isomorphism induced by $f$.
\end{lemma}

We also need the following standard group-theoretic lemmas. The first provides a key condition for identifying whether a non-Abelian simple group belongs in the socle. Namely, if $S_{1} \cong S_{2}$ are non-Abelian simple groups where $S_{1}$ is in the socle and $S_{2}$ is not in the socle, then the normal closures of $S_{1}$ and $S_{2}$ are non-isomorphic. In particular, the normal closure of $S_{1}$ is a direct product of non-Abelian simple groups, while the normal closure of $S_{2}$ is not a direct product of non-Abelian simple groups. We will apply this condition later when $S_{1}$ is a simple direct factor of $\Soc(G)$; in which case, the normal closure of $S_{1}$ is of the form $S_{1}^{k}$. We refer the reader to \cite[Section~6.1]{GLWL1} for the proofs.

\begin{lemma} \label{LemmaSocle}
Let $G$ be a finite semisimple group. A subgroup $S \leq G$ is contained in $\Soc(G)$ if and only if the normal closure of $S$ is a direct product of nonabelian simple groups.
\end{lemma}

\begin{proof}
See \cite[Lemma~6.5]{GLWL1}.
\end{proof}

\begin{corollary} \label{CorSocleFactor}
Let $G$ be a finite semisimple group. A nonabelian simple subgroup $S \leq G$ is a direct factor of $\Soc(G)$ if and only if its normal closure $N = ncl_G(S)$ is isomorphic to $S^k$ for some $k \geq 1$ and $S \unlhd N$.
\end{corollary}

\begin{proof}
See \cite[Corollary~6.6]{GLWL1}.
\end{proof}

\subsection{Main Result}

We first show that count-free WL can distinguish semisimple groups from those where $\rad(G)$ is non-trivial. In particular, this shows that count-free WL will distinguish direct products of finite non-Abelian simple groups from groups where $\rad(G)$ is non-trivial.

\begin{lemma} \label{lem:IsSemisimple}
Let $G$ be a semisimple group, and let $H$ be a group with Abelian normal subgroups. The count-free $(O(1), O(1))$-WL Version I will distinguish $G$ from $H$.
\end{lemma}

\begin{proof}
Let $M$ be an Abelian normal subgroup of $H$. Spoiler pebbles some element $h \in M$. Duplicator responds by pebbling some $g \in G$. Now as $M$ is Abelian and normal, $\text{ncl}_{H}(h) \leq M$ is Abelian. As $G$ does not have Abelian normal subgroups, $\text{ncl}_{G}(g)$ is non-Abelian.

Now recall that $\text{ncl}_{H}(h)$ is generated by the $H$-conjugates of $h$. As $\text{ncl}_{H}(h)$ is Abelian, the $H$-conjugates of $h$ all pairwise commute. On the other hand, there exist a pair of conjguates of $g$ that do not commute. Let $u, v \in G$ such that $ugu^{-1}, vgv^{-1}$ do not commute. Spoiler pebbles $u,v$ over the next two rounds. Regardless of what Duplicator pebbles, Spoiler wins with $O(1)$ additional pebbles and $O(1)$ additional rounds.
\end{proof}

Our next goal is to show that count-free WL can distinguish elements that belong to direct factors of $\Soc(G)$ from elements not in $\Soc(H)$. To do so, we will appeal to the \textit{rank lemma} \cite[Lemma~4.3]{GLWL1} of Grochow \& Levet. 

\begin{definition}[{\cite[Definition~4.1]{GLWL1}}]
Let $C \subseteq G$ be a subset of a group $G$ that is closed under taking inverses. We define the \emph{C-rank} of $g \in G$, denoted $\rk_C(g)$, as the minimum $m$ such that $g$ can be written as a word of length $m$ in the elements of $C$. If $g$ cannot be so written, we define $\rk_C(g) = \infty$.
\end{definition}

\begin{lemma}[{\cite[Rank Lemma~4.3]{GLWL1}}] \label{lem:RankLemma}
Suppose that $C \subseteq G$ is distinguished by the count-free $(k,r)$-WL Version I. If $\rk_{C}(g) \neq \rk_{C}(h)$, then Spoiler can win with $k+1$ pebbles and $\max\{r,\log d + O(1)\}$ rounds, where $d = \text{diam}(Cay(\langle C \rangle, C)) \leq |\langle C \rangle| \leq |G|$.
\end{lemma}

\begin{remark}
The way the rank lemma is formulated in \cite{GLWL1} relies on Duplicator selecting a bijection $f : G \to H$ to show that there exists an $x \in G$ such that $\rk(x) \neq \rk(f(x))$.  The remainder of the proof from \cite{GLWL1} goes through in the case of the count-free pebble game. We will use the rank lemma in cases where $g \mapsto h$ have been pebbled by assumption. Thus, the rank lemma will allow us to reason about count-free WL. The rank lemma in \cite{GLWL1} is also formulated in terms of WL Version II, but the argument goes through in the presence of WL Version I.
\end{remark}

\begin{lemma} \label{LemWeight1}
Let $G, H$ be a semisimple groups. Let $S \in \Fac(\Soc(G))$. If $g \in S$ and $h \not \in \Soc(H)$. The count-free $(O(1), O(\log \log n))$-WL Version I will distinguish $g$ from $h$.
\end{lemma}

\begin{proof}
We consider the pebble game, starting from $g \mapsto h$. Suppose first that $h$ is not contained in a subgroup $T \leq H$ that is isomorphic to $S$. In this case, Spoiler pebbles generators $s_{1}, \ldots, s_{7}$ for $S$ such that every element of $S$ can be written as a word of length $O(\log^{C} |S|)$ for some absolute constant $C$ over $s_{1}, \ldots, s_{7}$ (see \Thm{thm:BabaiKantorLubotzky}). Regardless of Duplicator's choice of $t_{1}, \ldots, t_{7}$, the map $(g, s_{1}, \ldots, s_{7}) \mapsto (h, t_{1}, \ldots, t_{7})$ does not extend to an isomorphism. So by the \textit{rank lemma} of Grochow \& Levet \cite[Lemma~4.3]{GLWL1} (recalled as \Lem{lem:RankLemma}), Spoiler wins with $O(1)$ additional pebbles and $O(\log \log n)$ additional rounds. 

So suppose now that there is a copy $T \leq H$ of $S$ containing $h$. Spoiler pebbles generators $s_{1}, \ldots, s_{7}$ for $S$ such that every element of $S$ can be written as a word of length $O(\log^{C} |S|)$ over $s_{1}, \ldots, s_{7}$ (see \Thm{thm:BabaiKantorLubotzky}). We may assume that Duplicator responds by pebbling generators $t_{1}, \ldots, t_{7}$ for $T$. Otherwise, by \Thm{thm:BabaiKantorLubotzky} and \cite[Lemma~4.3]{GLWL1}, Spoiler wins. 

As $h \not \in \Soc(H)$, we have that $T \not \in \Fac(\Soc(H))$. As $S \in \Fac(\Soc(G))$, the normal closure $\text{ncl}(S)$ is minimal normal in $G$ \cite[Exercise 2.A.7]{Isaacs2008FiniteGT}. As $h$ is not even contained in $\text{Soc}(H)$, we have by Lemma \ref{LemmaSocle} that $\text{ncl}(T)$ is not a direct product of non-Abelian simple groups, so $\text{ncl}(S) \not\cong \text{ncl}(T)$. We note that $\text{ncl}(S) = \langle \{ gSg^{-1} : g \in G \} \rangle$. 

As $\text{ncl}(T)$ is not isomorphic to a direct power of $S$, there are conjugates $u T u^{-1}, vTv^{-1} \neq T$ such that $uTu^{-1}, vTv^{-1}$ do not commute, by \Lem{LemmaSocle}. Yet since $S \unlhd \Soc(G)$, any two distinct conjugates of $S$ do commute. Spoiler pebbles such a $u, v \in H$, and Duplicator responds by pebbling some $u', v' \in G$. Regardless of Duplicator's choice of $u',v'$, the map $(s_{1}, \ldots, s_{7}, u', v') \mapsto (t_{1}, \ldots, t_{7}, u, v)$ does not extend to an isomorphism. By \Thm{thm:BabaiKantorLubotzky}, every element in $S$ can be written as a word of length $O(\log^{C} |S|)$ for some absolute constant $C$. Furthermore, as we need only check commutativity between $uTu^{-1}, vTv^{-1}$ and between $u'S(u')^{-1}, v'S(v')^{-1}$, we need only consider conjugates of elements in $S$ or $T$. So a distinguishing word has length at most $O(\log^{C} |S|) \leq O(\log^{C} n)$. Thus, by the \textit{rank lemma} \cite[Lemma~4.3]{GLWL1}, Spoiler wins with $O(1)$ additional pebbles and $O(\log \log n)$ additional rounds, as desired.
\end{proof}

\begin{remark} \label{rmk:ColoringSocleFactors}
Let $T \in \Fac(\Soc(G))$, and let $(t_{1}, \ldots, t_{7})$ be a generating sequence for $T$ prescribed by \Thm{thm:BabaiKantorLubotzky}, where every element of $T$ can be realized as a word of length $O(\log^{C} |T|)$ for some absolute constant $c$. Let $K \leq G$ be an isomorphic copy of $T$ such that $K \not \in \Fac(\Soc(G))$, and let $(k_{1}, \ldots, k_{7})$ be a generating sequence for $K$ such that the map $t_{i} \mapsto k_{i}$ for all $i \in [7]$ extends to an isomorphism. The proof of \Lem{LemWeight1} provides that the count-free $(O(1), O(\log \log n))$-WL Version I will assign different colors to $(t_{1}, \ldots, t_{7}, 1, \ldots, 1)$ and $(k_{1}, \ldots, k_{7}, 1, \ldots, 1)$.
\end{remark}

\begin{lemma} \label{LemIdentifySocle}
Let $G, H$ be a semisimple groups. Let $g \in \Soc(G)$ and $h \not \in \Soc(H)$, and suppose that $g \mapsto h$ has been pebbled. The count-free $(O(1), O(\log \log n))$-WL Version I will distinguish $g$ from $h$.
\end{lemma}

\begin{proof}
Define:
\[
C := \bigcup_{S \in \Fac(\Soc(G))} S.
\]

\noindent By \Lem{LemWeight1}, $C$ is identified by the count-free $(O(1), O(\log \log n))$-WL. Now $\Soc(G)$ has at most $O(\log |G|)$ non-Abelian simple direct factors. Thus, every element in $\Soc(G)$ can be written as a word of length $|\text{Fac}(\Soc(G))| \leq O(\log n)$ over $C$. So $g$ and $h$ will receive different colors in at most $O(\log \log n)$ rounds. 
\end{proof}

\begin{theorem}
Let $G$ be a direct product of non-Abelian simple groups, and let $H$ be arbitrary. We can decide whether $G \cong H$ in $\beta_{1}\textsf{MAC}^{0}(\textsf{FOLL})$.
\end{theorem}

\begin{proof}
By \Lem{lem:IsSemisimple}, we may assume that $H$ is semisimple; otherwise, the count-free $(O(1), O(1))$-WL will distinguish $G$ and $H$. As $G$ is a direct product of non-Abelian simple groups, we have that $G = \Soc(G)$. Now if $H \neq \Soc(H)$, then by \Lem{LemIdentifySocle}, Spoiler can win with $O(1)$ pebbles and $O(\log \log n)$ rounds.

We may now assume that $H$ is a direct product of non-Abelian simple groups. If $G \not \cong H$, then $\Fac(\Soc(G)) \neq \Fac(\Soc(H))$. In particular, there exists a non-Abelian direct factor $K$ of $\Soc(G)$ whose isomorphism class appears with greater multiplicity than in $\Fac(\Soc(H))$. By \Thm{thm:BabaiKantorLubotzky}, there exist generators $K = \langle x_{1}, \ldots, x_{7} \rangle$ such that every element of $K$ can be written as a word of length $O(\log^{C} |K|)$ (for some absolute constant $C$) over $x_{1}, \ldots, x_{7}$. 

We will use the count-free $(O(1), O(\log \log n))$-WL Version I, which by \Rmk{rmk:ColoringSocleFactors}, will assign two non-isomorphic finite simple groups different colors. Furthermore, by \Rmk{rmk:ColoringSocleFactors}, if $T = \langle t_{1}, \ldots, t_{7} \rangle \leq H$ is a copy of $K$ and $T \not \in \Fac(\Soc(H))$, then $(x_{1}, \ldots, x_{7}, 1, \ldots, 1)$ and $(t_{1}, \ldots,  t_{7}, 1, \ldots, 1)$ will receive different colors after $O(\log \log n)$ rounds. So in particular, any tuple $(t_{1}, \ldots, t_{7},1, \ldots, 1)$ over $H$ receiving the same color as $(x_{1}, \ldots, x_{7}, 1, \ldots, 1)$ after $O(\log \log n)$ rounds will correspond to an element of $\Fac(\Soc(H))$ such that the map $(x_{1}, \ldots, x_{7}) \mapsto (t_{1}, , \ldots, t_{7})$ extends to an isomorphism. Using $O(\log n)$ bits, we guess the color class corresponding to $(x_{1}, \ldots, x_{7}, 1, \ldots, 1)$. We use a single $\textsf{Majority}$ gate to compare the multiplicity of $\chi_{O(\log \log n)}((x_{1}, \ldots, x_{7}, 1, \ldots, 1))$ in $G$ to that in $H$. Precisely, for each tuple over $G$ that has the same color as $(x_{1}, \ldots, x_{7}, 1, \ldots, 1)$, we feed a $1$ to the \textsf{Majority} gate. And for each tuple over $H$ that has the same color $(x_{1}, \ldots, x_{7}, 1, \ldots, 1)$, we feed a $0$ to the $\textsf{Majority}$ gate. The parallel WL implementation due to Grohe \& Verbitsky \cite{GroheVerbitsky} stores indicators at each round as to whether two tuples receive the same color. So identifying the tuples that have the same color as $(x_{1}, \ldots, x_{7}, 1, \ldots, 1)$ is $\textsf{AC}^{0}$-computable. Thus, the entire procedure is computable using a $\beta_{1}\textsf{MAC}^{0}(\textsf{FOLL})$ circuit.

The result now follows.
\end{proof}

\section{Higher-arity Count-Free Lower Bound}

We generalize the count-free pebble game of Cai, F\"urer, \& Immerman \cite{CFI} to the $q$-ary setting. Fix $q \geq 1$. In the $q$-ary, $(k+1)$-pebble game, we again have two players: Spoiler and Duplicator. Spoiler wishes to show that the two groups $G$ and $H$ are non-isomorphic. Duplicator wishes to show that the two groups are isomorphic. Each round of the game proceeds as follows.
\begin{enumerate}
\item Spoiler picks up pebble pairs $(p_{i_{1}}, p_{i_{1}}'), \ldots, (p_{i_{d}}, p_{i_{d}}')$, where $1 \leq d \leq q$.

\item The winning condition is checked. This will be formalized later.

\item For each $j \in [d]$, Spoiler selects $j$ elements of one group on which to place the $j$ pebbles. After Spoiler has placed all their pebbles, Duplicator then responds by placing the corresponding pebbles on $j$ elements of the other group.
\end{enumerate}

\noindent \\ Let $g_{1}, \ldots, g_{m}$ be the pebbled elements of $G$ at the end of step $1$ of the given round, and let $h_{1}, \ldots, h_{m}$ be the corresponding pebbled elements of $H$. Spoiler wins precisely if the map $g_{i} \mapsto h_{i}$ does not extend to an isomorphism of $\langle g_{1}, \ldots, g_{m} \rangle$ and $\langle h_{1}, \ldots, h_{m} \rangle$. Duplicator wins otherwise. Spoiler wins, by definition, at round $0$ if $G$ and $H$ do not have the same number of elements. In the $q = 1$ case, the $(k+1)$-pebble, $r$-round count-free game is equivalent to the count-free $(k,r)$-WL (see \cite{CFI} for the setting of graphs and \cite{GLWL1} for the setting of groups).

We establish the following.

\begin{theorem}
Let $q \geq 1, n \geq 5$. Let $G_{n} := (\mathbb{Z}/2\mathbb{Z})^{qn} \times (\mathbb{Z}/4\mathbb{Z})^{qn}$ and $H_{n} := (\mathbb{Z}/2\mathbb{Z})^{qn-2} \times (\mathbb{Z}/4\mathbb{Z})^{qn+1}$. Duplicator has a winning strategy in the $q$-ary, $qn/4$-pebble game.
\end{theorem}

\begin{proof}
The proof is by induction on the number of pebbles. On the first round, Spoiler picks up $i \in [q]$ pebbles and places them on the board- without loss of generality, placing all pebbles on $G$. Let $(g_{1}, \ldots, g_{i})$ be the pebbled elements of $G$. Duplicator responds by placing $i$ pebbles on $(h_{1}, \ldots, h_{i})$. As $n \geq 5$ and $i \leq q$, Duplicator is able to pebble elements such that:
\begin{itemize}
\item The map $g_{j} \mapsto h_{j}$ extends to an isomorphism of $\langle g_{1}, \ldots, g_{i} \rangle \cong \langle h_{1}, \ldots, h_{i} \rangle \cong (\mathbb{Z}/2\mathbb{Z})^{a} \times (\mathbb{Z}/4\mathbb{Z})^{b}$, and

\item Exactly $0 \leq a_{1} \leq a$ of the $\mathbb{Z}/2\mathbb{Z}$ direct factors of $\langle g_{1} \ldots, g_{i} \rangle$ are contained in copies of $\mathbb{Z}/4\mathbb{Z}$ in $G_{n}$ if and only if exactly $a_{1}$ of the the $\mathbb{Z}/2\mathbb{Z}$ direct factors of $\langle h_{1} \ldots, h_{i} \rangle$ are contained in copies of $\mathbb{Z}/4\mathbb{Z}$ in $H_{n}$.
\end{itemize} 

\noindent \\ Now fix $1 \leq \ell < qn/4$, and suppose that Duplicator has a winning strategy in the $\ell$-pebble game. In particular, suppose that $g_{1}, \ldots, g_{m} \in G$ and $h_{1}, \ldots, h_{m} \in H$ have been pebbled, and that the following conditions hold:
\begin{itemize}
\item The map $\alpha : g_{j} \mapsto h_{j}$ extends to a marked isomorphism of $\hat{\alpha} : \langle g_{1}, \ldots, g_{i} \rangle \cong \langle h_{1}, \ldots, h_{i} \rangle$ of a subgroup of the form $(\mathbb{Z}/2\mathbb{Z})^{a} \times (\mathbb{Z}/4\mathbb{Z})^{b}$, and

\item Exactly $0 \leq a_{1} \leq a$ of the $\mathbb{Z}/2\mathbb{Z}$ direct factors of $\langle g_{1} \ldots, g_{i} \rangle$ are contained in copies of $\mathbb{Z}/4\mathbb{Z}$ in $G_{n}$ if and only if exactly $a_{1}$ of the the $\mathbb{Z}/2\mathbb{Z}$ direct factors of $\langle h_{1} \ldots, h_{i} \rangle$ are contained in copies of $\mathbb{Z}/4\mathbb{Z}$ in $H_{n}$.
\end{itemize} 

\noindent \\ As Duplicator has a winning strategy in the $\ell$-pebble game, it is not to Spoiler's advantage to move any existing pebbles. Thus, Spoiler picks up $p \in [q]$ new pebbles and places them on the board. If Spoiler places a given pebble on an element in $\langle g_{1}, \ldots, g_{m} \rangle$, then Duplicator may respond by placing corresponding pebble on the corresponding element of $\langle h_{1}, \ldots, h_{m} \rangle$. So we may assume that Spoiler places all $p$ pebbles on elements outside of $\langle g_{1}, \ldots, g_{m} \rangle$. Call these elements $g_{1}', \ldots, g_{p}'$. Denote $h_{1}', \ldots, h_{p}'$ to be the corresponding elements that Duplicator pebbles. We will show that Duplicator can select $h_{1}', \ldots, h_{p}'$ to win at this round. 

By similar argument as in the base case, Duplicator may select $h_{1}', \ldots, h_{p}'$ such that the map $g_{j}' \mapsto h_{j}'$ for all $j \in [p]$ extends to an isomorphism of $\langle g_{1}', \ldots, g_{p}' \rangle$ and $\langle h_{1}', \ldots, h_{p}' \rangle$. It remains to show that the map $(g_{1}, \ldots, g_{m}, g_{1}', \ldots, g_{p}') \mapsto (h_{1}, \ldots, h_{m}, h_{1}', \ldots, h_{p}')$ extends to an isomorphism of:
\[
\langle g_{1}, \ldots, g_{m}, g_{1}', \ldots, g_{p}' \rangle \text{ and } \langle h_{1}, \ldots, h_{i}, h_{1}', \ldots, h_{p}' \rangle.
\]

\noindent We now prove the inductive step, by induction on $j$, the number of new pebbles used. Precisely, we will show the following.

\begin{quote}
\textbf{Claim:} For each $1 \leq j \leq p$, we have that the map: $(g_{1}, \ldots, g_{m}, g_{1}', \ldots, g_{j}') \mapsto (h_{1}, \ldots, h_{m}, h_{1}', \ldots, h_{j}')$ extends to an isomorphism of:
\[
\langle g_{1}, \ldots, g_{m}, g_{1}', \ldots, g_{j}' \rangle \text{ and } \langle h_{1}, \ldots, h_{m}, h_{1}', \ldots, h_{j}' \rangle,
\]

\noindent which is a subgroup of the form $(\mathbb{Z}/2\mathbb{Z})^{a} \times (\mathbb{Z}/2\mathbb{Z})^{b}$. Furthermore, exactly $a_{1} \leq b_{j} \leq a$ of the $\mathbb{Z}/2\mathbb{Z}$ direct factors of $\langle g_{1} \ldots, g_{i} \rangle$ are contained in copies of $\mathbb{Z}/4\mathbb{Z}$ in $G_{n}$ if and only if exactly $b_{j}$ of the the $\mathbb{Z}/2\mathbb{Z}$ direct factors of $\langle h_{1} \ldots, h_{i} \rangle$ are contained in copies of $\mathbb{Z}/4\mathbb{Z}$ in $H_{n}$.

\begin{proof}[Proof of Claim]
The case of $j = 1$ is handled precisely by \cite[Theorem~7.10]{GLWL1}. Now fix $j \geq 1$, and suppose that the map $(g_{1}, \ldots, g_{m}, g_{1}', \ldots, g_{j}') \mapsto (h_{1}, \ldots, h_{m}, h_{1}', \ldots, h_{j}')$ extends to an isomorphism of:
\[
\langle g_{1}, \ldots, g_{m}, g_{1}', \ldots, g_{j}' \rangle \text{ and } \langle h_{1}, \ldots, h_{m}, h_{1}', \ldots, h_{j}' \rangle,
\]

\noindent which is a subgroup of the form $(\mathbb{Z}/2\mathbb{Z})^{a} \times (\mathbb{Z}/4\mathbb{Z})^{b}$. Furthermore, suppose that $a_{1} \leq b_{j} \leq a$ of the $\mathbb{Z}/2\mathbb{Z}$ direct factors of $\langle g_{1}, \ldots, g_{i}, g_{1}', \ldots, g_{j}' \rangle$ are contained in copies of $\mathbb{Z}/4\mathbb{Z}$ within $G_{n}$, and similarly $b_{j}$ of the $\mathbb{Z}/2\mathbb{Z}$ direct factors of $\langle h_{1}, \ldots, h_{i}, h_{i}', \ldots, h_{j}' \rangle$ are contained in copies of $\mathbb{Z}/4\mathbb{Z}$ within $H_{n}$.

Using the fact that $\ell \leq qn/4$, we may again apply \cite[Theorem~7.10]{GLWL1} to deduce that the map $(g_{1}, \ldots, g_{m}, g_{1}', \ldots, g_{j+1}') \mapsto (h_{1}, \ldots, h_{m}, h_{1}', \ldots, h_{j+1}')$ extends to an isomorphism of:
\[
\langle g_{1}, \ldots, g_{m}, g_{1}', \ldots, g_{j+1}' \rangle \text{ and } \langle h_{1}, \ldots, h_{m}, h_{i}', \ldots, h_{j+1}' \rangle,
\]

\noindent which is a subgroup of the form $(\mathbb{Z}/2\mathbb{Z})^{a'} \times (\mathbb{Z}/4\mathbb{Z})^{b'}$. We may further suppose that if exactly $b_{j} \leq b_{j+1} \leq a$ of the $\mathbb{Z}/2\mathbb{Z}$ direct factors of $\langle g_{1}, \ldots, g_{i}, g_{1}', \ldots, g_{j+1}' \rangle$ are contained in copies of $\mathbb{Z}/4\mathbb{Z}$ within $G_{n}$, then similarly exactly $b_{j+1}$ of the $\mathbb{Z}/2\mathbb{Z}$ direct factors of $\langle h_{1}, \ldots, h_{i}, h_{i}', \ldots, h_{j+1}' \rangle$ are contained in copies of $\mathbb{Z}/4\mathbb{Z}$ within $H_{n}$. The claim now follows by induction.
\end{proof}
\end{quote}

\noindent We note that while Duplicator must pebble its $m$ elements at once, Duplicator does so only after Spoiler selects its $p$ elements. So Duplicator may internally identify $h_{1}', \ldots, h_{p}'$ in sequence, selecting $h_{j}'$ to depend on $h_{1}', \ldots, h_{j-1}'$. Once Duplicator identifies the elements to pebble, Duplicator then places pebbles on $h_{1}', \ldots, h_{p}'$ at once. Thus, Duplicator has a winning strategy in the $\ell + q$ pebble game. The result now follows by induction.
\end{proof}

\section{Conclusion}

We used the count-free Weisfeiler--Leman in tandem with bounded non-determinism and an $\textsf{MAC}^{0}$ circuit to exhibit new upper bounds on isomorphism testing of several families of groups, including the CFI groups arising from Mekler's construction \cite{Mekler, HeQiao}, coprime extensions $H \ltimes N$ where $N$ is Abelian and $H$ is either (i) $O(1)$-generated solvable with solvability class $\poly \log \log n$ or (ii) finite simple, and direct products of simple groups. 


We also showed that the higher-arity count-free pebble game is unable to even distinguish  Abelian groups. This suggests that some measure of counting is necessary to place $\algprobm{GpI}$ into $\textsf{P}$.

Our work leaves several open questions.

\begin{question}
Suppose that $\Gamma_{0}$ is a graph identified by the count-free $3$-WL algorithm. Let $\Gamma_{1} = \text{CFI}(\Gamma_{0})$ and $\Gamma_{2} = \widetilde{\text{CFI}(\Gamma_{0})}$. Let $G_{i}$ ($i = 1, 2$) be the class $2$ $p$-group of exponent $p$ ($p > 2$) arising from $\Gamma_{i}$ via Mekler's construction \cite{Mekler, HeQiao}. Can the constant-dimensional count-free WL algorithm for groups distinguish $G_{1}$ from $G_{2}$?
\end{question}

As canonizing finite simple groups is $\textsf{FOLL}$-computable, the following problem seems within reach.

\begin{question} \label{question:Simple}
Let $G$ be a finite group given by its Cayley (multiplication) table. Can we decide whether $G$ is simple in $\textsf{FOLL}$? 
\end{question}

It is possible to decide in $\textsf{L}$ whether $G$ is simple by using a membership test. For each $x \in G$, we test whether the normal closure $\text{ncl}(x) = G$. It is not clear whether this membership test is $\textsf{FOLL}$-computable.

We also conjecture that isomorphism testing of groups without Abelian normal subgroups is decidable using $\beta_{1}\textsf{MAC}^{0}(\textsf{quasiFOLL})$-circuits of size $n^{\Theta(\log \log n)}$. The strategy from Grochow \& Levet \cite{GLWL1} was to individualize generators for the non-Abelian simple direct factors of $\Soc(G)$ and then apply the standard counting Weisfeiler--Leman algorithm. In order to compute direct the direct factors of $\Soc(G)$ and the relevant generators, Grochow \& Levet relied on a membership test, which is known to be $\textsf{L}$-computable \cite{BarringtonMcKenzie}. Resolving Question~\ref{question:Simple} is a first step towards resolving this conjecture.

Finally, we note that it is not clear as to the precise logic corresponding to the $q$-ary count-free game. Thus, we ask the following.

\begin{question}
For $q \geq 2$, determine the logic corresponding to the $q$-ary count-free game.
\end{question}

\bibliographystyle{alphaurl}
\bibliography{references}

\end{document}